\documentclass{article}
\usepackage{ijcai24}
\usepackage{times}
\usepackage{soul}
\usepackage{url}
\usepackage[utf8]{inputenc}
\usepackage{colorprofiles}
\usepackage[a-1b]{pdfx}
\hypersetup{pdfstartview=}

\usepackage[small]{caption}
\usepackage{graphicx}
\usepackage{amsmath}
\usepackage{amsthm}
\usepackage{booktabs}
\usepackage{algorithm}
\usepackage{algorithmic}
\usepackage[switch]{lineno}
\usepackage{thmtools}
\usepackage{thm-restate}
\usepackage{tikz-cd}
\usepackage{fancybox}
\usepackage{microtype}
\usepackage{tikz}
\usepackage{tikz-cd}
\usetikzlibrary{shapes.geometric}
\tikzset{
itria/.style={
  draw,shape border uses incircle,
  isosceles triangle,shape border rotate=90, yshift = -1.6cm, xshift = 0.7cm },
rtria/.style={
  draw,shape border uses incircle, minimum width = 1cm,
  minimum height = 1cm,
  isosceles triangle,isosceles triangle apex angle=130,
  shape border rotate=-65,yshift=-0.098cm,xshift=0.9346cm},
ritria/.style={
  draw,shape border uses incircle,
  isosceles triangle,isosceles triangle apex angle=110,
  shape border rotate=-55,yshift=0.1cm},
letria/.style={
  draw,dashed,shape border uses incircle,
  isosceles triangle,isosceles triangle apex angle=110,
  shape border rotate=235,yshift=0.1cm}
  
}

 \newcommand{\revision}[1]{#1}

\usepackage{mathtools}
\usepackage{amsmath,amsfonts, amsthm}
\usepackage{amssymb}
\usepackage{todonotes}
\usepackage{xspace}

\usepackage{tikz, tikz-qtree}

\usepackage{stmaryrd}
\newtheorem{theorem}{Theorem}[section]
\newtheorem{corollary}[theorem]{Corollary}
\newtheorem{lemma}[theorem]{Lemma}
\newtheorem{proposition}[theorem]{Proposition}
\newtheorem{definition}[theorem]{Definition}
\newtheorem{claim}
{Claim}

\newtheorem{exampl}{Example}[section]
\newenvironment{example}{\begin{exampl}\em}{\end{exampl}}

\newtheorem{rem}[theorem]{Remark}
\newenvironment{remark}{\begin{rem}\em}{\end{rem}}

\newcommand{\dlr}[1]{DL-Lite$^{#1}$}
\newcommand{\langfull}{\lang{\forall,\exists,\geq,-,\sqcap,\sqcup,\top,\bot,\neg}}
\newcommand{\lang}[1]{\mathcal{L}(#1)}
\newcommand{\ont}{\mathcal{O}}

\newcommand{\calI}{\Imc}

\usepackage{mysc}

\title{On the Power and Limitations of Examples for Description Logic Concepts}

\author{
Balder ten Cate$^1$
\and
Raoul Koudijs$^2$\and
Ana Ozaki$^{2,3}$
\affiliations
$^1$ Institute for Logic, Language and Computation (ILLC), University of Amsterdam\\
$^2$University of Bergen\\
$^3$University of Oslo
\emails
b.d.tencate@uva.nl, 
raoul.koudijs@uib.no,
anaoz@uio.no} 

\begin{document}

\maketitle

\begin{abstract}
Labeled examples (i.e., positive
and negative examples) are an 
attractive medium for communicating complex concepts. They are
useful for deriving concept
expressions (such as in concept learning, interactive concept specification, and
concept refinement) as well as for illustrating
concept expressions to a user or domain expert.
We investigate the power of labeled examples for describing
description-logic concepts. 
Specifically, we systematically
study the existence
and efficient computability of 
\emph{finite characterisations}, i.e., 
finite sets of labeled examples that
uniquely characterize a single concept,
for a wide variety of description logics
between $\EL$ and $\mathcal{ALCQI}$,both without an ontology and in the presence of a DL-Lite ontology. Finite characterisations are 
relevant for debugging purposes, and their existence is a necessary condition for exact learnability with membership queries.
\end{abstract}

\section{Introduction}
\thispagestyle{empty}

 Labeled examples (i.e., positive
and negative examples) are an 
attractive medium for communicating complex concepts. They are
useful as data for deriving concept
expressions (such as in concept learning, interactive concept specification, and example-driven
concept refinement) as well as for illustrating
concept expressions to a user or domain expert~\cite{DBLP:journals/jmlr/Lehmann09,DBLP:conf/ijcai/FunkJLPW19,FJL-IJCAI21,DBLP:journals/ki/Ozaki20,DBLP:conf/ilp/FanizzidE08,DBLP:journals/apin/IannonePF07,DBLP:journals/ml/LehmannH10}. 
Here, we study the utility of labelled examples for \revision{characterising
description logic concepts, where examples are finite interpretations that 
can be either positively or negatively labelled.}

\begin{example}
    In the description logic $\EL$,
    we may define the concept of an \emph{e-bike} by means of the concept expression 
    \[ C \text{ ~ := ~ } \text{Bicycle}\sqcap\exists \text{Contains}.\text{Battery}\] 
    Suppose we wish
    to illustrate $C$ by a
    collection of positive and negative
    examples. What would be a good choice of examples?
    Take the interpretation $\Imc$ consisting of the following facts.

    \begin{center}
    \cornersize{.2}\Ovalbox{\cornersize{1}\begin{tabular}{r@{~}c@{~}l}
    Bicycle  \ovalbox{soltera2} & $\xrightarrow{\text{Contains}}$ & \ovalbox{li360Wh} Battery \\[3mm]
    Bicycle  \ovalbox{px10} & $\xrightarrow{\text{Contains}}$ & \ovalbox{b12} Basket \\[3mm]
    Car \ovalbox{teslaY} & $\xrightarrow{\text{Contains}}$ &  \ovalbox{li81kWh}  Battery  
    \end{tabular}}
    \end{center}
    In the context of this interpretation $\Imc$, it is clear that
    \begin{itemize}
        \item \emph{soltera2} is a positive example for $C$, and
        \item \emph{px10} and \emph{teslaY} are negative examples for $C$
    \end{itemize}
    In fact, as it turns out, $C$ is the \emph{only} $\EL$-concept
    (up to equivalence) that fits these three labeled examples. In other words, these three labeled examples ``uniquely characterize'' $C$
    within the class of all $\EL$-concepts. This shows that the above
    three labeled examples are a good choice of examples. Adding any
    additional examples would be redundant. Note, however, that this 
    depends on the choice of description logic. For instance,  the richer concept language $\ALC$ allows for other concept expressions such as
    $\text{Bicycle}\sqcap \neg\exists \text{Contains}.\text{Basket}$ that also fit.
\end{example}

Motivated by the above example, we 
investigate the existence and efficient computability of 
\emph{finite characterisations}, i.e., 
finite sets of labeled examples that
uniquely characterize a single concept.
Finite characterisations are 
relevant not only for \emph{illustrating} a complex concept to a user (e.g., to verify the correctness of a concept expression obtained using machine learning techniques). Their existence is a necessary condition for \emph{exact learnability with membership queries}~\cite{DBLP:journals/ml/Angluin87}. Furthermore, from a more fundamental point of view, by studying the existence of finite characterisations, we gain insight into the power and limitations of labeled examples as a medium for describing concepts.

Concretely, we systematically study the existence of, and the complexity of computing, finite characterisations for description logic concepts in a wide range of expressive description logics. We look at concept languages $\lang{\Obf}$ that are generated by a set of connectives $\Obf$, where $\{\exists,\sqcap\}\subseteq\Obf\subseteq\{\forall,\exists,\geq,-,\sqcap,\sqcup,\top,\bot,\neg\}$. 
In other words, we look at  fragments of the description logic $\mathcal{ALCQI}$ that contain at least the 
$\exists$ and $\sqcap$ constructors from $\EL$.


As our first result, 
within this framework, we obtain an almost complete classification of the concept languages that 
admit finite characterisations.



\begin{figure}[t]
\includegraphics[scale=0.9]{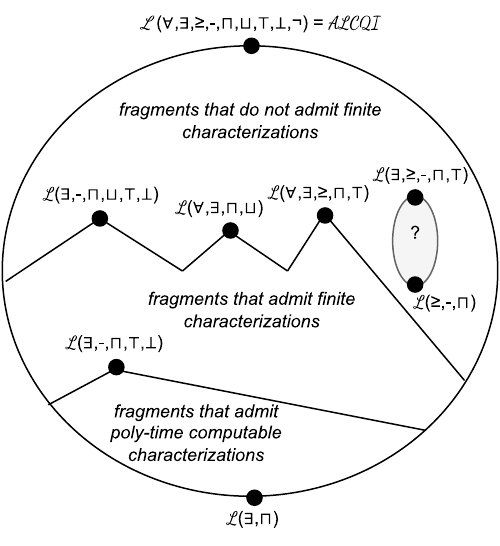}
\vspace{-3mm}
\caption{Summary of Thm.~\ref{thm:mainone} and Thm.~\ref{thm:maintwo}.}
\label{fig:main}
\end{figure}

\begin{restatable}{theorem}{mainone}
\label{thm:mainone}
\mbox{Let $\{\exists,\sqcap\}\subseteq \Obf\subseteq\{\forall,\exists,\geq,-,\sqcap,\sqcup,\top,\bot,\neg\}$.}
\begin{enumerate}
    \item 
If $\Obf$ is \revision{a subset of}  $\{\exists,-,\sqcap,\sqcup,\top,\bot\}$, $\{\forall,\exists,\sqcup,\sqcap\}$ or $\{\forall,\exists,\geq,\sqcap,\top\}$ then $\lang{\Obf}$ admits finite characterisations. 
\item Otherwise $\lang{\Obf}$ does not admit finite characterisations, except possibly if
 $\{\geq,-,\sqcap\}\subseteq\Obf\subseteq\{\exists,\geq,-,\sqcap,\top\}$. 
\end{enumerate}
    
\end{restatable}

The above theorem identifies three maximal fragments of $\mathcal{ALCQI}$ that admit finite characterisations. It leaves open the status of essentially only two concept
languages, namely $\lang{\exists,\geq,-,\sqcap}$ and $\lang{\exists,\geq,-,\sqcap,\top}$ (note that $\exists$ is definable in terms of $\geq$). The proof builds on prior results from \cite{BalderCQ} and \cite{TOCLarxiv}. Our main novel technical contributions are a construction of finite characterisations for $\lang{\forall,\exists,\geq,\sqcap,\top}$ and complementary negative results for $\lang{\geq,\bot},\lang{\geq,\sqcup}$ and $\lang{\forall,\exists,-,\sqcap}$.
  The construction of finite characterisations for $\lang{\forall,\exists,\geq,\sqcap,\top}$ is non-elementary. We give an elementary (doubly exponential) construction for $\lang{\exists,\geq,\sqcap,\top}$, via a novel polynomial-time algorithm for constructing frontiers.

Next, we study which concept languages $\lang{\Obf}$ \emph{admit poly\-nomial-time computable characterisations}, i.e., have a poly\-nomial-time algorithm to compute finite characterisations. 

\begin{restatable}{theorem}{maintwo}
\label{thm:maintwo}
\mbox{Let $\{\exists,\sqcap\}\subseteq \Obf\subseteq\{\forall,\exists,\geq,-,\sqcap,\sqcup,\top,\bot,\neg\}$.} 
\begin{enumerate}
    \item If $\Obf$ is \revision{a subset of}  $\{\exists,-,\sqcap,\top,\bot\}$, then $\lang{\Obf}$ admits polynomial-time computable characterisations.
    \item  Otherwise, $\lang{\Obf}$ does not admit polynomial-time computable characterisations, assuming P$\neq$NP.
\end{enumerate}
    
\end{restatable}

The first item follows from known results~\cite{BalderCQ}; we prove the second.
Thm.~\ref{thm:mainone} and~\ref{thm:maintwo} are summarized in Figure~\ref{fig:main}.

Finally, we investigate finite characterisations relative to an ontology $\ont$, 
where we require the examples to be interpretations that satisfy the ontology
and we require that every fitting concept is equivalent to the input concept $C$ relative to $\ont$. 

\begin{restatable}{theorem}{mainthree}
\label{thm:mainthree}
\mbox{Let $\{\exists,\sqcap\}\subseteq \Obf\subseteq\{\forall,\exists,\geq,-,\sqcap,\sqcup,\top,\bot,\neg\}$.}
\begin{enumerate}
    \item If $\Obf$ is \revision{a subset of}  $\{\exists,\sqcap,\top,\bot\}$, then $\lang{\Obf}$ admits finite characterisations
w.r.t.~all DL-Lite  ontologies.
\item Otherwise, $\lang{\Obf}$ does not admit finite characterisations w.r.t.~all DL-Lite   ontologies, except possibly if $\{\exists,\sqcap,\sqcup\}\subseteq\Obf\subseteq\{\exists,\sqcap,\sqcup,\top,\bot\}$.
\end{enumerate}
In fact, for $\lang{\exists,\sqcap,\top,\bot}$-concepts $C$ and DL-Lite ontologies $\ont$ such that $C$ is satisfiable w.r.t.~$\ont$, a finite characterisation can be computed in polynomial time.
\end{restatable}

This shows in particular that, if we exclude  $\sqcup$ from consideration, then $\lang{\exists,\sqcap,\top,\bot}$ is the unique largest fragment admitting finite characterisations under DL-Lite ontologies.
\thispagestyle{empty}


\paragraph{Outline}
In Section~\ref{sec:prelims}, we review relevant definitions and we introduce our framework for finite characterisations. 
Section~\ref{sec:withoutontologies} studies the existence and polynomial-time computability of finite characterisations without an ontology. Section~\ref{sec:withontologies} extends the analysis to the case with DL-Lite ontologies. Section~\ref{sec:conclusion} discusses
further implications and future directions.

\paragraph{Related Work}
Finite characterisations were first studied in the computational learning theory literature under the name of teaching sets, with a corresponding notion of teaching dimension, measuring the maximal size of minimal teaching sets of some class of concepts \cite{GOLDMAN199520}. The existence of finite characterisations is a necessary condition for  exact learnability with membership queries.


Several recent works study finite characterisations for \emph{description logic concepts} (\cite{BalderCQ,DuplicateDBLP:conf/ijcai/FunkJL22} for $\mathcal{ELI}$; \cite{FortinKR2022}
for temporal instance queries). We use
results from~\cite{BalderCQ,DuplicateDBLP:conf/ijcai/FunkJL22}.
Thm.~\ref{thm:mainthree} 
 appears to contradict a result in~\cite{DuplicateDBLP:conf/ijcai/FunkJL22}, which 
states that $\mathcal{ELI}$ admits 
finite characterisations under DL-Lite ontologies. However, this is due to a difference in the way we
define examples, which we will discuss in detail in Section~\ref{sec:prelims}.

Finite characterisations have also been studied for a while in the \emph{data management} community
(cf.~\cite{Mannila1985TestDF,Staworko2015CharacterizingXT,BalderCQ,ICDT2024} for queries; \cite{Alexe2011CharacterizingSM} for schema mappings), and,
a systematic study of finite characterisations for syntactic fragments of \emph{modal logic} was carried out in~\cite{TOCLarxiv}.  We make use of several
results from~\cite{TOCLarxiv}, and
one of our results also implies an answer to an 
open problem from this paper.


\section{Basic Definitions and Framework}
\label{sec:prelims}

In the following, we first define the syntax and semantics of
concept and ontology languages considered in this work.

\paragraph{Concept Languages}
Let $\NC$ and $\NR$ 
be infinite and mutually disjoint sets of \emph{concept} and \emph{role} 
names, respectively. We work with finite subsets $\Sigma_{\mathsf{C}}\subseteq\NC,\Sigma_{\mathsf{R}}\subseteq\NR$ of these. We denote by $\langfull[\Sigma_{\mathsf{C}},\Sigma_{\mathsf{R}}]$ the following concept language:
\begin{align*}
    C,D ::= & A \mid \top \mid \bot \mid (C\sqcap D) \mid (C\sqcup D) |\\
    & \neg C\mid (\exists S.C) \mid (\forall S.C) \mid (\geq k S. C)
\end{align*}
where $k\geq 1$, 
$A\in\Sigma_{\mathsf{C}}$ and $S=R$ or $S=R^-$ for some $R\in\Sigma_{\mathsf{R}}$. Note that $\exists S. C$ is equivalent to $\geq 1 S. C$.  

For any set of operators $\Obf\subseteq \{\forall,\exists,\geq,-,\sqcap,\sqcup,\top,\bot,\neg\}$, we denote by $\lang{\Obf}[\Sigma_{\mathsf{C}},\Sigma_{\mathsf{R}}]$ the fragment of $\langfull$ that only uses the constructs in $\Obf$, concept names from $\Sigma_{\mathsf{C}}$ and role names from $\Sigma_{\mathsf{R}}$. We may omit $[\Sigma_{\mathsf{C}},\Sigma_{\mathsf{R}}]$ from $\lang{\Obf}[\Sigma_{\mathsf{C}},\Sigma_{\mathsf{R}}]$ and simply write $\lang{\Obf}$ if the role and concepts names are clear from context or irrelevant. 

Observe that $\langfull$ is a notational variant of 
$\mathcal{ALCQI}$, while, e.g.,
$\lang{\exists,\sqcap,\top}$ corresponds
to $\mathcal{EL}$ and $\lang{\forall,\exists,\sqcap,\top,\bot}$ 
to $\mathcal{ALE}$. This notation thus allows us to state results for a large variety of languages for concepts.

\paragraph{Ontology Language} 
We consider the very popular
 DL-Lite  ontology language~\cite{dllite-jair09}.
 \dlr{}  concept
inclusions (CIs) are expressions of the form $B \sqsubseteq C$,
respectively, where $B$, $C$ are
concept expressions built through the grammar rules (we write below $\exists S$ as a shorthand for $\exists S.\top$)
\[B ::= A \mid \exists R\mid \exists R^-, ~~~~~ C ::= B \mid \neg B,\]
with $R \in \NR$ and $A \in \NC$.
We call concepts of the form $B$ above \emph{basic concepts}. An \emph{DL-Lite ontology} is a finite set 
of DL-Lite CIs. Note that every DL-Lite ontology is satisfiable.

\paragraph{Semantics}
An \emph{interpretation} $\calI=(\Delta^\Imc,\cdot^\Imc)$ w.r.t. a signature $(\Sigma_{\mathsf{C}},\Sigma_{\mathsf{R}})$ is a
structure, in the traditional model-theoretic sense. It consists of a non-empty set $\Delta^{\Imc}$, called the \emph{domain}, and a function $\cdot^{\Imc}$ that assigns a subset $A^{\Imc}\subseteq\Delta^{\Imc}$ of the domain to each $A\in\Sigma_{\mathsf{C}}$ and a binary relation  $R^{\Imc}\subseteq \Delta^{\Imc}\times\Delta^{\Imc}$ over the domain to each $R\in\Sigma_{\mathsf{R}}$.
We extend $(\cdot)^\Imc$ to  concept and role expressions as follows. 
\begin{align*} 
	(\top)^\Imc := {} & \Delta^\Imc \quad (\bot)^\Imc := {}  \emptyset\quad  (C\sqcap D)^\Imc := {} C^\Imc\cap D^\Imc\\ 
   (\neg C)^\Imc := {} &  \Delta^\Imc\setminus C^\Imc
     \qquad \quad\quad  (C\sqcup D)^\Imc := {} C^\Imc\cup D^\Imc \\
	    (\exists S.C)^\Imc := {} & \{d\mid \exists e\in\Delta^\Imc \text{ s.t. }(d,e)\in S^\Imc \text{ and } e\in C^\Imc\}\\
	    (\geq k S.C)^\Imc := {} & \{d\mid \exists e_1, \ldots, e_k\in\Delta^\Imc\;\text{pairwise-distinct s.t.} \\& (d,e_i)\in S^\Imc \text{ and } e_i\in C^\Imc\;\text{for all}\;1\leq i\leq k\}\\
    (\forall S.C)^\Imc := {} & \{d\mid \forall e\in\Delta^\Imc
    \text{ if }(d,e)\in S^\Imc \text{ then } e\in C^\Imc\}\\
    (R^-)^\Imc := {} & \{(e,d)\mid (d,e)\in R^\Imc\}.
  \end{align*} 
We say that an interpretation $\Imc$ satisfies a CI $C\sqsubseteq D$ 
if $C^\Imc\subseteq D^\Imc$. We write $\Imc\models\alpha$ if $\Imc$ satisfies a CI $\alpha$. Moreover, we say that $\Imc$ satisfies an ontology $\ont$, written  $\Imc\models\ont$, if $\Imc\models\alpha$
for all CIs $\alpha$ in $\ont$. 
A concept $C$ is \emph{satisfiable} w.r.t.~an ontology $\ont$ (or vice-versa) if there is an interpretation $\Imc\models\ont$ with $C^\Imc\neq\emptyset$; a concept $C$ \emph{subsumes} $D$ relative to $\ont$
(denoted $\ont\models C\sqsubseteq D$)
if $C^\Imc\subseteq D^\Imc$ for all
interpretations $\Imc\models\ont$; and
two concepts $C,D$  are \emph{equivalent} relative to $\ont$ if $C^{\Imc} = D^{\Imc}$ for all interpretations $\Imc\models\ont$. We   use $\models C\sqsubseteq D$ as a shorthand for $\emptyset\models C\sqsubseteq D$ when indicating that a subsumption holds relative to the empty ontology. We also write $C\equiv D$ as a shorthand
for $C\equiv_\emptyset D$.

A \emph{pointed interpretation} is a pair $(\Imc,d)$ where $\Imc$ is an interpretation and $d\in\Delta^\Imc$ is a domain element. 
We say that a pointed interpretation $(\Imc,d)$ \emph{satisfies} a concept $C$ if $d\in C^{\Imc}$. 
A pointed interpretation $(\Imc,d)$ is 
 \emph{finite} if 
$\Delta^\Imc$ is a finite set.

Given a pointed interpretation $(\Imc,d)$ for $(\Sigma_{\mathsf{C}},\Sigma_{\mathsf{R}})$, the \emph{height} of $(\Imc,d)$ \emph{w.r.t. a role} $R\in\Sigma_{\mathsf{R}}$ (also denoted  $h_{R}(\Imc,d)$) is the maximum $n$ such that there is a path
\[ d=d_1\xrightarrow{R} d_2 \cdots \xrightarrow{R} d_n\]
in $\Imc$ starting at $d$, or $\infty$ if there is no maximum such number. Note that for all pointed interpretations $(\Imc,d)$, role names $R\in\Sigma_{\mathsf{R}}$, and  $n\geq 0$, we have $h_{R}(\Imc,d)=n$ iff $d\in (height^R_n)^{\Imc}$, where $height^{R}_{n}$ is the $\lang{\forall,\exists,\sqcap,\bot}$-concept $[\forall R.]^{n+1}\bot\sqcap[\exists R.]^{n}\forall R.\bot$.


\paragraph{Size of Concepts and Examples}
For any concept $C\in\lang{\forall,\exists,\geq,-,\sqcap,\sqcup,\top,\bot}$, we write $|C|$ for the number of symbols in the syntactic expression $C$. The numbers occurring in $C$ are assumed to be represented
using a binary encoding. All our lower bounds also hold when numbers are represented in unary.
Let $\mathsf{nr}(C)$ be the the largest natural number occurring in $C$. 
Further, let $\mathsf{dp}(C)$ be the \emph{role  depth} of $C$, i.e. the longest sequence of nested quantifiers and/or $\geq$ operators in $C$. Finally, for interpretations $\Imc$ we set $|\Imc|:=|\Delta^{\Imc}|$, and the size of an
example $(\Imc,d)$ is simply $|\Imc|$.

\begin{restatable}{theorem}{finitecontrollability}(Finite Controllability)
\label{thm:finitecontrollability}
Let $C,D\in\lang{\forall,\exists,\geq,-,\sqcap,\sqcup,\top,\bot,\neg}$ and let $\ont$ be a DL-Lite ontology. 
Then $C\not\equiv_\Omc D$ iff there is a finite pointed interpretation 
$(\Imc,d)$ with $\Imc\models\ont$ and  $d\in((C\sqcap\neg D)\sqcup(\neg C\sqcap D))^\Imc$.
\end{restatable}

Let $\mathcal{L}$ be a concept language. Then the model checking problem for $\mathcal{L}$ is the problem of determining, on input a pointed interpretation $(\Imc,d)$ together with a $\mathcal{L}$-concept, whether $d\in C^{\Imc}$. Measuring the \emph{combined} complexity of this problem means that we parameterize the complexity of whether $d\in C^{\Imc}$ on both $|C|$ and $|\Imc|$.

\begin{theorem}\label{thm:modelchecking}
Model checking $\mathcal{ALCQI}$ is in polynomial time for combined complexity.
\end{theorem}

\begin{proof}
Straightforward bottom-up evaluation
 solves the problem in time $O(|C|\cdot |\Imc|)$. This procedure consists of enumerating all (at most $|C|$ many) subconcepts
of the input concept, in order of increasing complexity, and, 
for each such subconcept $C'$, computing (in time linear in $|\Imc|$) its
interpretation $(C')^{\Imc}$ using 
the fact that the interpretation $(C'')^\Imc$ of 
all proper subconcepts $C''$ of $C'$
has already been computed.
\end{proof}

\paragraph{Labeled Examples and Finite characterisations}

\thispagestyle{empty}
An \emph{example} is a finite pointed  interpretation $(\Imc,d)$,  labelled either positively (with a ``$+$'') or negatively (with a ``$-$''). We say that an example $(\Imc,d)$ is a
\emph{positive example} for a concept $C$  if $d\in C^\Imc$, and otherwise $(\Imc,d)$ is a \emph{negative example} for $C$. We may  omit ``for $C$'' when this is clear from the context. We will use the notation $E=(E^+,E^-)$ to denote
a finite collection of examples labeled as positive
and negative. A concept $C$ \emph{fits} a set of positively and negatively labelled examples if it is satisfied on all the positively labelled examples and not satisfies on any of the negatively labelled examples.\footnote{Equivalently, a concept $C$ fits a set of positively and negatively labelled examples if all the positively labelled examples are positive examples for $C$, and all the negatively labelled examples are negative examples for $C$.}

A \emph{finite characterisation} for a concept $C$ w.r.t.~concept language $\mathcal{L}$ is  a finite collection of labeled examples  $E=(E^+,E^-)$ such that $C$ fits $E$ and every concept $D\in\mathcal{L}$ that fits $E$ is equivalent to $C$. A concept language $\mathcal{L}$ \emph{admits finite characterisations} if for every pair of finite sets $\Sigma_{\mathsf{C}}\subseteq\NC,\Sigma_{\mathsf{R}}\subseteq\NR$, every concept $C\in\mathcal{L}[\Sigma_{\mathsf{C}},\Sigma_{\mathsf{R}}]$ has a finite characterisation w.r.t.~$\mathcal{L}[\Sigma_{\mathsf{C}},\Sigma_{\mathsf{R}}]$. We say that $\mathcal{L}$ \emph{admits polynomial-time computable characterisations} if, furthermore, there is an algorithm that, given an input concept $C\in\mathcal{L}[\Sigma_{\mathsf{C}},\Sigma_{\mathsf{R}}]$ outputs a finite characterisation $E=(E^+,E^-)$ of $C$ w.r.t.~$\mathcal{L}[\Sigma_{\mathsf{C}},\Sigma_{\mathsf{R}}]$ in time polynomial in the size of the concept $|C|$ and the size of the signature $|\Sigma_{\mathsf{C}}|+|\Sigma_{\mathsf{R}}|$. This implies that $E^+$ and $E^-$ consist of polynomially many, polynomial-size examples. \revision{Therefore, if $\lang{\Obf}$ admits polynomial time computable characterisations it has polynomial \emph{teaching size} in the sense of \cite{Telle-Orallo2019Teachingsize}}. \revision{Admitting finite or polynomial time computable characterisations is a monotone property of concept classes.} 

We defined examples as \emph{finite}
pointed interpretations, because
they are intended as objects that 
can be communicated to a user and/or specified by a user. This is justified by Thm.~\ref{thm:finitecontrollability}, which 
shows that any two non-equivalent concepts can be distinguished by a finite pointed interpretation.

It may not be clear to the reader why the above definition of ``admitting finite characterisations'' involves a quantification over vocabularies. The following example explains this.

\begin{example}
    Consider the concept
    $C=\exists R.A$. Let $E=(E^+,E^-)$ with $E^+=\{(\Imc, d_1)\}$ and $E^-=\{(\Imc,d_3)\}$, where $\Imc$
    is the following
    interpretation.

    \begin{center}
    \cornersize{.2}\Ovalbox{
      \cornersize{1}
      \begin{tabular}{r@{~}c@{~}l@{~}c@{~}l@{~}l@{~}l@{~}l}
         \ovalbox{$d_1$} & $\xrightarrow{ ~~R~~ }$ & \ovalbox{$d_2$}$^A$ & ~~~~
       $^A$ \ovalbox{$d_3$} & $\xrightarrow{ ~~R~~ }$ & \ovalbox{$d_4$} $\xrightarrow{ ~~R~~ }$ & \ovalbox{$d_5$}$^A$ \rotatebox{90}{$\circlearrowleft$} $^R$
      \end{tabular}
    }
    \end{center}
    
    Clearly, $C$ fits $E$. In fact, 
    it can be shown that $C$ is the only
    $\lang{\exists,\sqcap,\sqcup}$-concept (up to equivalence)
    that fits $E$ and that uses only the
    concept name $A$ and role name $R$.
    In other words, $E$ is a finite characterisation of $C$
    w.r.t. $\lang{\exists,\sqcap,\sqcup}[\Sigma_{\mathsf{C}},\Sigma_{\mathsf{R}}]$
    with $\Sigma_{\mathsf{C}}=\{A\}$ and $\Sigma_{\mathsf{R}}=\{R\}$.     
    In contrast, $E$ does \emph{not} uniquely characterize $C$ with respect to $\lang{\exists,\sqcap,\sqcup}[\Sigma'_C,\Sigma_{\mathsf{R}}]$ with $\Sigma'_C=\{A,B\}$ since $\exists R.(A\sqcup B)$ fits $E$ and is not equivalent to $C$.
    Therefore, what qualifies as a finite characterisations depends on the choice of vocabulary.
\end{example}

The notions of an \emph{example}
and of a \emph{finite characterisation}
naturally extend to the case with an 
ontology $\ont$, where we now require that each
example satisfies $\ont$,
and we replace \emph{concept equivalence} by \emph{equivalence relative to $\Omc$}.
More precisely, an \emph{example for $\ont$} is a finite pointed interpretation $(\Imc,d)$ with $\Imc\models\ont$, labelled either positively or negatively. 
A \emph{finite characterisation} for a concept $C$ w.r.t.~a concept language $\mathcal{L}$ \emph{under $\ont$} is a finite collection $E=(E^+,E^-)$ of examples for $\ont$ such that for every $D\in\mathcal{L}$ that fits $E$,
$C\equiv_\ont D$. 
{
A concept language $\mathcal{L}$ \emph{admits finite characterisations under $\mathcal{L'}$  ontologies} if for every 
  $\ont$ in $\mathcal{L'}$ and every pair of finite sets $\Sigma_{\mathsf{C}}\subseteq\NC,\Sigma_{\mathsf{R}}\subseteq\NR$, every concept $C\in\mathcal{L}[\Sigma_{\mathsf{C}},\Sigma_{\mathsf{R}}]$ 
has a finite characterisation w.r.t. $\mathcal{L}[\Sigma_{\mathsf{C}},\Sigma_{\mathsf{R}}]$ under $\ont$. }
\begin{example}\label{ex:catdog}
Consider the concept language
$\lang{\sqcap,\sqcup,\neg}$ and
concept names $\Sigma_{\mathsf{C}}=\{\text{Animal}, \text{Cat}, \text{Dog}, \text{Red}\}$. Note that there are no
role constructs in the language.
Let $\ont$ be the DL-Lite ontology
that expresses that Cat and Dog are disjoint 
subconcepts of Animal. 
Now consider
the concept $C=\text{Cat}\sqcap\text{Red}$.
Let $\Imc$ be the interpretation with 
\begin{itemize}
\item $\Delta^{\Imc}=\{a,a',b,b',c,c',d,d'\}$, 
\item $\text{Animal}^\Imc=\{a,a',c,c',d,d'\}$, 
\item $\text{Cat}^\Imc=\{c,c'\}$,
$\ \text{Dog}^\Imc=\{d,d'\}$, $\ \text{Red}^\Imc=\{a,b,c,d\}$,
\end{itemize}
which satisfies $\ont$.
It can be shown that the collection of labeled examples $(E^+,E^-) $ with $E^+=\{(\Imc,c)\}$ and $E^-=\{(\Imc,a),(\Imc,a'), (\Imc,b), (\Imc,b'), (\Imc,c'),(\Imc,d), (\Imc,d')\}$ is a finite
characterisation of $C$ w.r.t. $\lang{\sqcap,\sqcup,\neg}$ under $\Omc$. 
On the other hand, the same set of 
examples is not a finite characterisation
under the empty ontology, because
$\text{Cat}\sqcap\text{Red}\sqcap\neg\text{Dog}$ also fits but is not equivalent to $C$
in the absence of the ontology.
Thus, in this case, the ontology helps to reduce the number of examples needed
to uniquely characterize $C$. Moreover, the ontology can rule out unnatural examples, since in this case the characterisation would need to include a negative example satisfying $\text{Cat}\sqcap\text{Red}\sqcap\text{Dog}$.
\end{example}
\begin{remark}\label{rem:exampledef}
    We have defined an example as a pointed finite interpretation $(\Imc,d)$ and we call such an example a \emph{positive example} for $C$ if $d\in C^\Imc$ and a \emph{negative example} otherwise. In some of the prior literature on learning description logic concepts, examples were instead defined as pointed ABoxes $(\Amc,a)$~\cite{DuplicateDBLP:conf/ijcai/FunkJL22,DBLP:conf/aaai/KonevOW16,DBLP:conf/aaai/OzakiPM20} \footnote{\revision{By a pointed ABox $(\Amc,a)$ we mean a finite set $\Amc$ of   assertions of the form $A(b)$ and $R(b,c)$ with $A$ a concept name,  $R$ a role name, and $b,c$ individual names (over the given signature), together with an individual name $a$ occurring in some assertion in $\Amc$.}}\revision{, while e.g. \cite{FortinKR2022} also work with finite interpretations (in their case, linear orders) as examples.} Such a pointed ABox is then considered a \emph{positive example} for $C$ under an ontology $\ont$ if
    $a$ is a certain answer for $C$ in $\Amc$ under $\ont$ (i.e., $a^\Imc\in C^\Imc$ for all interpretations with $\Imc\models\ont$ and $\Imc\models\Amc$) 
    , and a \emph{negative example} otherwise.
    Let us call the latter type of examples ``ABox-examples''. 
\thispagestyle{empty}
    Consider the concepts
    $C = \forall R . P$ and $C'=\forall R.Q$. Although they are clearly not equivalent, it is not possible to distinguish them by means of ABox-examples (under, say, an empty ontology), for the simple reason that
    there does not exist a
    positive ABox-example for either of these concepts. 
    This motivates our decision to define examples as pointed interpretations. This is tied to the fact that we study  concept languages that include universal quantification and/or negation. \revision{Hence, for concept languages contained in $\lang{\exists,\geq,\sqcap,\sqcup,\top,\bot}$, such as $\EL$, there is essentially no difference between finite characterisations using ABox-examples and finite characterisations using interpretation-examples in the absence of an ontology.}

    However, there are some further repercussions for having  interpretations instead of ABoxes as examples in the case with ontologies, even for such weaker concept languages. \revision{Intuitively, in the presence of a DL-Lite$^{\Hmc}$-ontology $\ont$, any ABox consistent with $\ont$ can be viewed as a succinct representation of a possibly infinite interpretation (i.e. its \emph{chase}).} Conversely, a finite interpretation can always be conceived of as an ABox. The difference
shows up as follows:
it was shown in~\cite{DuplicateDBLP:conf/ijcai/FunkJL22}
that $\lang{\exists,\sqcap,-}$ admits finite
characterisations under DL-Lite$^\Hmc$ ontologies 
in the setting where examples are pointed ABoxes. We see in Section~\ref{sec:withontologies} 
below 
that the same does not hold in our setting. 
\end{remark}

\section{Finite Characterisations without Ontology}
\label{sec:withoutontologies}

Our aim, in this section, is to prove
Thm.~\ref{thm:mainone} and~\ref{thm:maintwo}. We 
start by listing relevant known results.

\begin{theorem}\label{thm:knownresults}
The following results are known
(cf.~\cite{TOCLarxiv}).
\begin{enumerate}
\item
    $\lang{\exists,-,\sqcap,\top,\bot}$ admits
    polynomial-time computable characterisations (from \cite{BalderCQ}).
\item
    $\lang{\exists,-,\sqcap,\sqcup,\top,\bot}$ admits
    finite characterisations. However, finite characterisation are necessarily of exponential size (from \cite{Alexe2011CharacterizingSM}).
\item
    $\lang{\forall,\exists,\sqcap,\sqcup}$ admits
    finite characterisations.  However, finite characterisation are necessarily of non-elementary size (from \cite{TOCLarxiv}).
\item 
    Neither $\lang{\forall,\exists,\sqcap,\bot}$ nor $\lang{\forall,\exists,\sqcup,\top}$ admit finite characterisations (from \cite{TOCLarxiv}).
\end{enumerate}
\end{theorem}

To establish Thm.~\ref{thm:mainone},
it remains to show that $\lang{\forall,\exists,\geq,\sqcap,\top}$ admits finite characterisations and that $\lang{\geq,\bot},\lang{\geq,\sqcup}$ and $\lang{\forall,\exists,-,\sqcap}$ do not admit finite characterisations.

\subsection{Finite Characterisations for \texorpdfstring{$\lang{\forall,\exists,\geq,\sqcap,\top}$}{L(forall,exists,geq,and,top)} }

The next results involve a method for constructing examples that forces fitting concepts to have a bounded role depth. 

\begin{restatable}{proposition}{knrdepthn}
\label{prop:knrdepthn}
\revision{Fix finite sets $\Sigma_C\subseteq \NC$ and 
$\Sigma_R\subseteq \NR$.
For all $n,k\geq 0$, there is a set of examples $E_{n,k}^+$, such that
    for all $\lang{\forall,\exists,\geq,\sqcap,\top}[\Sigma_C,\Sigma_R]$ concept expressions $C$,    
    the following are equivalent:}
    \begin{enumerate}
    \item $C$ is equivalent to a concept expression of role depth at most $n$ and maximum number restriction $k$;
    \item $C$ fits the positive examples $E^+_{n,k}$.  \end{enumerate}
    
\end{restatable}

\revision{
We illustrate the result for the 
special case where  $n=3$ and $k=2$. 
It suffices to choose as positive
examples the examples $(\Imc,d_1)$ and $(\Imc',d_1)$ based on the 
following interpretations:}
\[ \Imc:~~~~
\begin{tikzcd}[row sep = -2mm]
     & d^1_2 \arrow[r,""] \arrow[rdd,""] & d^1_3 \arrow[rd,""]& \\
d_1 \arrow[ru,""]\arrow[rd,""] &       &       & e & ~~ \\
     & d^2_2 \arrow[r,""] \arrow[ruu,""]& d^2_3 \arrow[ru,""]& 
\end{tikzcd}
\]
\[ \Imc':~~~~
\begin{tikzcd}[row sep = -2mm]
     & d^1_2 \arrow[r,""] \arrow[rdd,""] & d^1_3 \arrow[rd,""]& \\
d_1 \arrow[ru,""]\arrow[rd,""] &       &       & e \arrow[r,""] \arrow[loop,looseness=3] & e'\\
     & d^2_2 \arrow[r,""] \arrow[ruu,""]& d^2_3 \arrow[ru,""]& 
\end{tikzcd}
\]
\revision{where an arrow $d\to d'$ means that the
pair $(d,d')$ belong to the interpretation of
every role in $\Sigma_R$; and every element except
$e'$ belongs to the interpretation of every atomic concept in
$\Sigma_C$.
It can be shown that every concept fitting $(\Imc,e)$ and $(\Imc',e)$ as positive is equivalent to a depth 0 concept. It follows by backward induction
that
an $\lang{\forall,\exists,\geq,\sqcap,\top}[\Sigma_C,\Sigma_R]$-concept fits
$(\Imc,d_1)$ and $(\Imc',d_1)$ as positive
examples iff it has role depth
at most 3 and maximum number restriction 2.
}

As a consequence, we obtain the following.

\begin{restatable}{theorem}{ALEQcharacterisations}
\label{thm:ALEQcharacterisations}
$\lang{\forall,\exists,\geq,\sqcap,\top}$ admits finite characterisations.
\end{restatable}

\revision{We show this by appending to $E^+_{n,k}$ as many positive and negative examples needed to distinguish a concept from all the finitely many other concepts with depth at most $n$ and number restriction at most $k$. Hence, we obtain} a non-elementary upper bound on the size of finite characterisations for $\lang{\forall,\exists,\geq,\sqcap,\top}$, since there are already $\text{tower}(n)$ (i.e. an iterated power $2^{\ldots^2}$ of height $n$) many pairwise non-equivalent $\lang{\exists,\sqcap,\top}$ concepts of role depth $n$. We do not have a matching lower bound for the above language, and it remains open whether there are more efficient approaches. In the next subsection, we show that the fragment $\lang{\exists,\geq,\sqcap,\top}$ \emph{does} admit an elementary (doubly exponential) construction of finite characterisations. 

\subsection{An Elementary Construction of Finite Characterisations for \texorpdfstring{$\lang{\exists,\geq,\sqcap,\top}$}{L(exists,geq,and,top)}}

Since $\exists$ is expressible as 
$\geq 1$, we restrict our attention to 
$\lang{\geq,\sqcap,\top}$ in what follows.

\begin{restatable}{theorem}{polytimesubsumption}
\label{thm:polytimesubsumption}
There is a polynomial time algorithm for testing if $\models C\sqsubseteq C'$, where $C,C'\in \lang{\geq,\sqcap,\top}$.
\end{restatable}

\revision{The above theorem can be proved by showing that entailments between conjunctions of $\lang{\geq,\sqcap,\top}$-concepts can be reduced to entailments between the individual conjuncts.}


Next, we will use Theorem~\ref{thm:polytimesubsumption} to 
show that concepts in $\lang{\geq,\sqcap,\top}$ can be brought into a certain normal form in polynomial time.
Recall that every $\lang{\geq,\sqcap,\top}$-concept expression $C$ that is not a concept name is a conjunction $C_1\sqcap \cdots \sqcap C_n$ ($n\geq 0$) where each
$C_i$ is either an concept name or is of
the form
$\geq k R . C'$ where $C'$ is again a $\lang{\geq,\sqcap,\top}$-concept expression. Here we conveniently view
$\top$ as a shorthand for an empty conjunction.
We say that such a concept expression $C$ 
is \emph{irredundant} if it does not contain a 
conjunction that includes distinct conjuncts
$C_i, C_j$ with $\models C_i\sqsubseteq C_j$. 

\thispagestyle{empty}

\begin{restatable}{proposition}{propnormalformpolytime}
\label{prop:normalformpolytime}
Every $\lang{\geq,\sqcap,\top}$-concept  $C$ is equivalent to an irredundant $\lang{\geq,\sqcap,\top}$-concept, which can be 
computed in polynomial time.
\end{restatable}

The following two lemmas, intuitively, tell us that, every pointed interpretation satisfying a concept
$C\in \lang{\geq,\sqcap,\top}$ contains a 
small pointed sub-interpretation that already satisfies $C$ \revision{and falsifies all $\lang{\geq,\sqcap,\top}$-concepts falsified in the original interpretation}. We write $\Imc\subseteq\Imc'$ to denote that $\Imc$ is a sub-interpretation of $\Imc'$ (that is, $\Delta^\Imc\subseteq\Delta^{\Imc'}$
and $X^\Imc\subseteq X^{\Imc'}$ for all 
$X\in \Sigma_C\cup\Sigma_R$).
\begin{restatable}{lemma}{smallsubinterpretation}
\label{lem:small-subinterpretation}
For every concept $C\in \lang{\geq,\sqcap,\top}$, interpretation $\Imc$, and $d\in C^\Imc$, there is an $\Imc'\subseteq \Imc$ with $d\in C^{\Imc'}$ and $|\Imc'|\leq |C|^{|C|}$.
\end{restatable}

Lem.~\ref{lem:small-subinterpretation} pairs well with
the next lemma, showing that 
$\lang{\geq,\sqcap,\top}$-concepts are $\subseteq$-monotone.

\begin{lemma}\label{lem:subset-monotone}
    For all concepts $C\in \lang{\geq,\sqcap,\top}$ and interpretations $\Imc\subseteq \Imc'$, it holds that $C^\Imc\subseteq C^{\Imc'}$.
\end{lemma}

This can be proved by a straightforward induction argument, or, alternatively, follows from the \L os-Tarski theorem.

\revision{It follows that, in order
to construct a finite characterisation of an $\lang{\geq,\sqcap,\top}$-concept, it suffices to pick as our positive examples 
all positive examples of size at most 
$|C|^{|C|}$. It only remains to explain how we construct the negative
examples. This is done through a novel \emph{frontier} construction construction for $\lang{\exists,\geq,\sqcap,\top}$.} Intuitively, a frontier for $C$ is finite set of concepts that are strictly weaker than $C$ and that separates $C$ from all concepts strictly weaker than $C$ 
(cf.~Figure~\ref{fig:frontier} for a graphical depiction).

\begin{definition}{(Frontier)}\label{def:frontier}
Given a concept language $\mathcal{L}$ and a concept $C\in\mathcal{L}$, a \emph{frontier} for $C$ w.r.t. $\mathcal{L}$ is a finite set of concepts $\{C_1,\ldots,C_n\}\subseteq\mathcal{L}$ such that:
\begin{enumerate}
    \item[(i)] $\models C\sqsubseteq C_i$ and $\not\models C_i\sqsubseteq C$ for all $1\leq i\leq n$, and
    \item[(ii)] for all $D\in\mathcal{L}$ with $\models C\sqsubseteq D$ and $\not\models D\sqsubseteq C$ we have $\models C_i\sqsubseteq D$ for some $1\leq i\leq n$.
\end{enumerate}
Frontiers \emph{relative to an ontology $\ont$} are 
defined similarly, except that 
each 
$\models$ or $\not\models$ is replaced by ``$\ont\models$'', resp.~``$\ont\not\models$''.
\end{definition}

 \begin{figure}
 \begin{center}
 \begin{tikzpicture}[level distance=25pt,sibling distance=6pt,grow'=up]
 \Tree [.\node(top){} ; \edge[draw=none];
 [.\text{$C$}
 \edge[-] ; [.\node(f1){$C_1$}; ]
 \edge[draw=none] ; [.\text{\dots} \edge[draw=none] ; [.\node(bot){$\top$}; ] ]
 \edge[-] ; [.\node(fn){$C_n$}; ] ]
 ]
 \draw (top) edge[out=180,in=180,looseness=1.6] (bot.west);
 \draw (top) edge[out=0,in=0,looseness=1.6] (bot.east);
\draw (f1) edge[out=-270,in=-140] (bot.west);
 \draw (fn) edge[out=-270,in=-40] (bot.east);
 \draw[dashed] (-1.2,1.5) rectangle (1.2,2.1);
 \end{tikzpicture}
 \end{center}
 \vspace{-3mm}
 \caption{A frontier in the concept subsumption hierarchy}
 \label{fig:frontier}
 \end{figure}
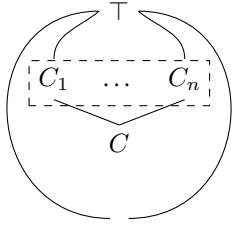

\begin{restatable}{theorem}{frontier}
\label{thm:frontier}
For each $C\in\lang{\geq,\sqcap,\top}$, one can construct a frontier $\mathcal{F}(C)$ for $C$ w.r.t. $\lang{\geq,\sqcap,\top}$ in polynomial time.
\end{restatable}

\begin{restatable}{theorem}{ELQcharacterisations}
\label{thm:ELQcharacterisations}
$\lang{\geq,\sqcap,\top}$ admits finite characterisations of doubly exponential size.
\end{restatable}

Thm.~\ref{thm:ELQlowerbound} below implies an exponential lower bound for the same problem, leaving us with a gap of one exponential.

\begin{remark}
\label{rem:refinement}
   The polynomial-time frontier construction from Thm.~\ref{thm:frontier} is interesting in its own right, particulary because 
frontiers are closely related to \emph{upward refinement operators}, a core component of 
DL learning systems such as the ELTL variant of
 DL Learner~\cite{DBLP:journals/ws/BuhmannLW16}. 
We can show that the function $\rho$
mapping each $\lang{\geq,\sqcap,\top}$-concept
to its frontier is an \emph{ideal upward refinement operator} as defined in~\cite{Lehmann2007FoundationsOR}.
This follows
from the fact that, for every  $\lang{\geq,\sqcap,\top}$-concept $C$ there 
are only finitely many $\lang{\geq,\sqcap,\top}$-concepts $D$, up to equivalence, for which $\models C\sqsubseteq D$. The latter, in turn, holds because $\models C\sqsubseteq D$ implies that $\mathsf{dp}(D)\leq \mathsf{dp}(C)$ and $\mathsf{nr}(D)\leq\mathsf{nr}(C)$.
\end{remark}

\thispagestyle{empty}

\subsection{Negative Results}

The following theorem summarizes our negative
result pertaining to the existence of finite
characterisations.

\begin{restatable}{theorem}{countingcupbot}
\label{thm:countingcup}
\label{thm:countingbot}
\label{thm:inversecap}
$\lang{\geq,\sqcup}$, $\lang{\geq, \bot}$ and $\lang{\forall,\exists,-,\sqcap}$ 
do not admit finite characterisations.
\end{restatable}

In the remainder of this section, we establish
lower bounds on the size and/or the complexity of computing
finite characterisations. The following
lemma will be helpful for this.

\begin{restatable}{lemma}{lemcanposexample}
\label{lem:canposexamples}
Let $\sqcap\in\Obf$ and $(E^+,E^-)$ be a finite characterisation of a concept $C\in\lang{\Obf}$ w.r.t. $\lang{\Obf}$. For all $C'\in\lang{\Obf}$, if $\not\models C\sqsubseteq C'$ then $C'$ is falsified on some positive example in $E^+$. 
\end{restatable}

Using this lemma, we show:

\begin{restatable}{theorem}{ELQlowerbound}
\label{thm:ELQlowerbound}\
$\lang{\geq,\sqcap}$ does not admit polynomial size characterisations. In fact, there are $\lang{\geq,\sqcap}$-concepts $(C_i)_{i\in\mathbb{N}}$ of constant depth and number restriction such that every finite characterisation of $C_i$ has size at least $2^{|C_i|}$.
\end{restatable}

The next lemma establishes a connection between the problem of computing finite characterisations and the problem of checking subsumptions between concepts.
\begin{restatable}{lemma}{lempolycharimpliessubsumptioninp}
\label{lem:polychar_implies_subsumptioninP}
If $\lang{\Obf}$ admits polynomial characterisations, $\sqcap\in\Obf$ and $\lang{\Obf}$ has a polynomial time model checking problem (in combined complexity), then subsumptions between $\lang{\Obf}$ concepts can be tested in polynomial time.
\end{restatable}

 This allows us to use hardness results from the literature on subsumption algorithms to establish hardness of computing certain example sets.
Note that 
the converse of Lem.~\ref{lem:polychar_implies_subsumptioninP} fails, since by Theorem~\ref{thm:polytimesubsumption} $\lang{\geq,\sqcap,\top}$ has a polynomial time subsumption problem but does not admit
polynomial-time computable characterisations
(Thm.~\ref{thm:ELQlowerbound}).

The following subsumption hardness result follows
from~\cite{DONINI1992}. 

\begin{restatable}{proposition}{NPhardness}
\label{prop:NPhardness}
Subsumption checking for $\lang{\forall,\exists,\sqcap}$ is $\mathsf{NP}$-hard.
\end{restatable}
\noindent
Using Lem.~\ref{lem:polychar_implies_subsumptioninP} 
we can the establish our next result, resolving an open question from \cite{TOCLarxiv}. 

\begin{corollary}\label{cor:ALE_NPhardness}
$\lang{\forall,\exists,\sqcap}$ does not admit polynomial time computable characterisations, unless $\mathsf{P}=\mathsf{NP}$.
\end{corollary}

Theorems~\ref{thm:mainone} and~$\ref{thm:maintwo}$ follow 
from the results proved in this section and prior results that we recalled in Thm.~\ref{thm:knownresults}.   


\section{Finite Characterisations With Ontologies}
\label{sec:withontologies}

In this section we are concerned with finite characterisations of concepts w.r.t. ontologies.

\subsection{\texorpdfstring{$\lang{\exists,\sqcap,\top,\bot}$}{L(exists,and,top,bot)} Admits Finite Characterisations Under DL-Lite Ontologies}\label{sec:el-dllite}



We show that $\lang{\exists,\sqcap,\top,\bot}$ admits 
finite characterisations under DL-Lite  ontologies. 
%
 %
 The proof uses the following key lemma from~\cite{DuplicateDBLP:conf/ijcai/FunkJL22}.
 Recall the definition of a frontier relative to an ontology
 (Def.~\ref{def:frontier}).

 \begin{lemma} (\cite{DuplicateDBLP:conf/ijcai/FunkJL22})
 \label{lem:frontiers-el-dllite}
  There is a polynomial-time algorithm that takes
  as input a $\lang{\exists,\sqcap,\top}$-concept $C$ and a DL-Lite ontology $\ont$ such 
  that $C$ is satisfiable w.r.t.~$\ont$, and outputs a frontier 
  for $C$ w.r.t.~$\lang{\exists,\sqcap,\top}$ relative to $\ont$.
\end{lemma}

Our second ingredient is a polynomial-time \emph{canonical
model construction} for $\lang{\exists,\sqcap,\top}$-concepts w.r.t.~DL-Lite ontologies. 
Intuitively, one can think of a canonical model,
for a concept $C$ and ontology $\ont$, as a
minimal positive example, i.e., a
positive example for $C$ under $\ont$ that
``makes as few other concepts true as possible''. We show that every $\lang{\exists,\sqcap,\top}$ concept that is satisfiable w.r.t. a given DL-Lite ontology $\ont$ admits a polynomial-time computable canonical model w.r.t.~$\ont$ that satisfies its CIs.


\begin{restatable}{theorem}{thmcanmod}
\label{thm:canmod}
There is a polynomial-time algorithm that takes a $\lang{\exists,\sqcap,\top}$-concept $C$
    and a DL-Lite ontology $\ont$,
    where $C$ is satisfiable w.r.t.~$\ont$, and that 
    produces a pointed interpretation 
    $(\Imc_{C,\ont},d_C)$
    (which we will call the \emph{canonical model} of $C$ under $\ont$) with the following
    properties.
    \begin{enumerate}
    \item for all $\lang{\exists,\sqcap,\top}$-concepts $D$, $d_C\in D^{\Imc_{C,\ont}}$
        iff $\ont\models C\sqsubseteq D$. In particular, we have that $d_C\in C^{\Imc_{C,\ont}}$.
        \item Moreover, $\Imc_{C,\ont}\models\ont$. 
    \end{enumerate}  
\end{restatable}

The proof is based on previous constructions of canonical model for query answering~\cite{DBLP:conf/kr/KontchakovLTWZ10}.
Canonical models for query answering usually differ from those for satisfying the ontology. In the case of the \EL ontology language (consisting of CIs $C\sqsubseteq D$ where $C,D\in\lang{\exists,\sqcap,\top}$), it was shown that a polynomial canonical model for \EL concepts can be constructed \cite{DBLP:conf/ijcai/BaaderBL05,DBLP:journals/jsc/LutzW10}, but it cannot be ensured that this canonical model satisfies the ontology.  In fact there are \EL ontologies for which there is no \emph{finite} canonical model satisfying the ontology~\cite[Example 33]{DBLP:conf/nmr/0001OR23}.


Thm.~\ref{thm:canmod} \emph{does} produce 
models satisfying the DL-Lite ontology, and this is important for
our purpose.
By combining the frontier construction with the 
canonical model construction, we can construct
finite characterisations in polynomial time
for concepts that are satisfiable w.r.t.~a
given ontology. For inconsistent concepts, we 
separately prove the following.

\begin{restatable}{lemma}{lembottomcharacterisation}
\label{lem:bottomcharacterisationontology}
Let $\ont$ be a DL-Lite ontology. Then the unsatisfiable concept $\bot$ admits an exponentially sized finite characterisation w.r.t. $\lang{\exists,\sqcap,\top,\bot}$ under $\ont$. \end{restatable}

\noindent
Given the above, we obtain the 
main result of this section.

\begin{restatable}{theorem}{thmpositiveontology}
\label{thm:EL-poly-char-DLLite}
$\lang{\exists,\sqcap,\top,\bot}$ admits finite
characterisations under DL-Lite ontologies. Moreover, for concepts $C\in\lang{\exists,\sqcap,\top,\bot}$ that are satisfiable w.r.t.~the given ontology $\ont$,
a finite characterisation can be computed
in polynomial time (in $|C|$ and $|\ont|$).
\end{restatable}

\thispagestyle{empty}

\begin{remark}
    The restriction to satisfiable concepts 
    in the above theorem is essential given that the ontology is treated as part of the input.
    Specifically, let $\ont$ be the ontology
    $\ont:=\{A_0\sqsubseteq \neg A_0\}\cup\{A_i\sqsubseteq\neg B_i \mid i=1\ldots n\}$. It can be shown, for the concept $A_0$ (which is unsatisfiable w.r.t.~$\ont$), that
    every finite characterisation must be of
    size at least $2^n$ (see e.g.~\cite{FJL-IJCAI21} for a similar style of counterexample). If the ontology is treated as fixed in the complexity analysis, however, a finite characterisation could be computed in polynomial time for \emph{all} 
    $\lang{\exists,\sqcap,\top,\bot}$-concepts. 
\end{remark}




\subsection{Negative Results}

The following two negative results establish   that we cannot further
enrich the concept language.

\begin{restatable}{theorem}{thmcountingontologies}
\label{thm:countingontologies}
\label{thm:ALEontologies}
Let $\ont$ be the DL-Lite ontology $\{A\sqsubseteq \neg A\}$. Then $\lang{{\geq}}$ and $\lang{\forall,\exists,\sqcap}$ do not admit finite characterisations under $\ont$.
\end{restatable}

\begin{restatable}{theorem}{inverseontology}
\label{thm:inverseontology}
Let $\ont$ be the DL-Lite ontology
$\{A\sqsubseteq\exists R,$ $ \exists R^-\sqsubseteq A\}$.
Then $\lang{\exists,-,\sqcap}$ does not admit finite
characterisations under $\ont$.
\end{restatable}

These results, together with Thm.~\ref{thm:EL-poly-char-DLLite}, imply Thm.~\ref{thm:mainthree}, as they show that no extension of $\lang{\exists,\sqcap,\top,\bot}$ with $\geq$ or $\forall$ admits finite characterisations under DL-Lite ontologies.
Our positive result in Section~\ref{sec:el-dllite} also cannot be extended to DL-Lite$^\Hmc$ ontologies (that is, DL-Lite extended with role inclusions~\cite{dllite-jair09}). This contrasts with a recent positive result for \EL with DL-Lite$^\Hmc$ ontologies using ABox-examples~\cite{DuplicateDBLP:conf/ijcai/FunkJL22}.

\begin{restatable}{theorem}{thmnegeldlliteh}
\label{thm:inverseontologyel}
Let $\ont$ be the DL-Lite$^\Hmc$ ontology
$\{A\sqsubseteq\exists R, \exists R^-\sqsubseteq A, R^-\sqsubseteq S\}$. Then
$\lang{\exists,\sqcap}$ does not admit finite
characterisations under $\ont$.
\end{restatable}

\section{Summary and Discussion}
\label{sec:conclusion}

We systematically studied the 
existence and complexity of computing finite characterisations for
concept languages $\lang{\Obf}$ with $\{\exists,\sqcap\}\subseteq\Obf\subseteq\{\forall,\exists,\geq,-,\sqcap,\sqcup,\top,\bot,\neg\}$, both in the absence of an ontology and in the presence of a DL-Lite ontology.
 While we classified most concept
 languages, there are a few cases left open,
 detailed below.

In passing, we obtained a polynomial-time
algorithm for checking subsumptions between $\lang{\exists,\geq,\sqcap,\top}$ concepts, and a polynomial time algorithm for computing frontiers of
$\lang{\exists,\geq,\sqcap,\top}$-concepts
(which gives rise to a polynomial-time ideal upward refinement
operator, cf.~Remark~\ref{rem:refinement}).


\paragraph{Repercussions for Exact Learnability}
As we mentioned in the introduction, the existence of polynomial-time computable characterisations is a necessary precondition for the existence of polynomial-time exact learning algorithms with membership queries. 
It was shown in~\cite{BalderCQ} that 
$\lang{\exists,-,\sqcap,\top,\bot}$ is indeed polynomial-time exactly learnable with membership queries. 
 Thm.~\ref{thm:maintwo} therefore implies that $\lang{\exists,-,\sqcap,\top,\bot}$ is a maximal fragment of $\mathcal{ALCQI}$ that is efficiently exactly learnable with membership queries, in the absence of an ontology. 

Conversely, every algorithm for effectively computing finite characterisations also gives rise to an effective exact learning algorithm with membership queries, although not an efficient one. 
Thus, 
Thm.~\ref{thm:mainone} also tells us something about the limits of (not-necessarily efficient) exact learnability with membership queries, without an ontology.

Finally, in the same way, Thm.~\ref{thm:mainthree} tells us something about the limits of exact learnability with membership queries where the membership queries can only be asked for examples satisfying a given DL-Lite ontology.
The authors are not aware of prior work on such a notion of learnability.



\paragraph{Open questions}
 Open questions include
 whether $\lang{\exists,\geq,-,\sqcap}$ or $\lang{\exists,\geq,-,\top}$ admit finite characterisations (under the empty ontology) and whether concept languages between $\lang{\exists,\sqcap,\sqcup}$ and $\lang{\exists,\sqcap,\sqcup,\top,\bot}$ admit finite characterisations under DL-Lite ontologies. 
Moreover, we do not yet have tight bounds on the complexity of computing finite characterisations for many of the concept languages we study that admit finite characterisations. In particular, for concept languages between $\lang{\forall,\exists,\sqcap}$ and $\lang{\forall,\exists,\geq \sqcap,\top}$ we have an NP lower bound (from Corollary~\ref{cor:ALE_NPhardness}) and a non-elementary upper bound (from Thm.~\ref{thm:ALEQcharacterisations}). In Thm.~\ref{thm:ELQcharacterisations} we managed to improve this to a doubly exponential upper bound for  $\lang{\exists,\geq,\sqcap,\top}$,
  while an exponential lower bound for $\lang{\exists,\geq,\sqcap}$ follows from Thm.~\ref{thm:ELQlowerbound}. It remains 
  to find tight bounds for these fragments.
 
 Since our setting includes $\forall$ and $\neg$, we were forced to work with finite pointed interpretations as examples (cf.~Remark~\ref{rem:exampledef}). If one discards these connectives it becomes natural to adopt pointed ABoxes as examples, as in \cite{DuplicateDBLP:conf/ijcai/FunkJL22}.
 Many of our negative results for the existence of finite characterisations under DL-Lite ontologies rely on examples being finite pointed interpretations satisfying the ontology, except Thm.~\ref{thm:countingontologies} because for the ontology in question, there is essentially no difference between a (consistent) ABox and a finite interpretation. 
 Our study therefore only leaves open the possibility of positive results in the ABox-as-examples setting under DL-Lite ontologies for fragments of
 $\lang{\exists,-,\sqcap,\sqcup,\top,\bot}$.

  \thispagestyle{empty}
 It would be of interest to investigate finite characterisations under 
  ontology languages other than DL-Lite.
 The recent paper \cite{ICDT2024} implies some positive results for $\lang{\exists,\sqcap,\sqcup,\top,\bot}$ relative to ontologies defined by certain universal Horn formulas (e.g. symmetry and transitivity). 

One of the motivations of our study was to generate examples to illustrate a concept to a user (e.g, for debugging purposes or for educational purposes). Finding practical solutions for this require more research, e.g., on how to generate ``natural'' examples (cf.~\cite{Glavic2021Book} for an overview of 
related problems and techniques in data management).



\thispagestyle{empty}

\section*{Acknowledgments}

\thispagestyle{empty}
Balder ten Cate is supported by the European Union's Horizon
2020 research and innovation programme under grant MSCA-101031081. 
Raoul Koudijs and Ana Ozaki are supported by the Research Council of Norway, project number $316022$, led by Ana Ozaki.  

\bibliographystyle{named}
\bibliography{bib}




\appendix
\section{Additional Preliminaries}
\label{sec:additional}

We will use the following abbreviations for writing complex nested concept. Given a concept $C$, we use  $[\geq k S.]^n C$ as a shorthand
notation for
\[\underbrace{\geq k S.\geq k S.\cdots.\geq k S.}_{\text{$n$ times}} C\]
and we use  $[\forall S.]^n C$ as a shorthand
notation for
\[\underbrace{\forall S.\forall S.\cdots.\forall S.}_{\text{$n$ times}} C\]
Formally, $[\geq k S.]^0C:=C$ and $[\geq k S.]^{n+1}.C := ~\geq k S.([\geq k S.]^{n}.C)$, where $k,n$ are natural numbers and $S$ is a role name from $\NR$ or the inverse of a role name in $\Sigma_{\mathsf{R}}$ inverse. Similarly, we set $[\forall S.]^0C:=C$ and $[\forall S.]^{n+1}C:=\forall S.[\forall S.]^nC$ for every role name $S$ or the inverse of a role name. 
Given a role name $R\in\NR$, we use
$\overline{R}$ to toggle between 
a role name and its inverse.
That is, if $R\in\NR$ then we set  $\overline{R}=R^-$ and $\overline{R^-}=R$.

To simplify our proofs, we assume that all ontologies
are in \emph{named form}, meaning that
we only allow CIs where
one of the sides is a concept name.
Every DL-Lite ontology $\ont$ 
can be converted in polynomial time to a conservative extension $\ont'$ of \Omc in named form (over an extended signature)  
 ~\cite{DBLP:conf/rweb/Krotzsch12}. Let $\mathsf{sig}(\ont)$ denote the finite set of concept and role names occurring in $\ont$. 

\begin{lemma}\label{lem:namedform}
Every DL-Lite ontology $\ont$ may be transformed in polynomial time into a DL-Lite ontology $\ont^*$ in named form such that:
\begin{itemize}
    \item $\mathsf{sig}(\ont^*)\supseteq\mathsf{sig}(\ont)$
    \item for all CIs $B\sqsubseteq C$ over $\mathsf{sig}(\ont)$, $\ont\models B\sqsubseteq C$ iff $\ont^*\models B\sqsubseteq C$
    \item For every concept language $\mathcal{L}$, $\mathcal{L}$ admits finite (polynomial time computable) characterisations under $\ont$ iff $\mathcal{L}$ admits finite (polynomial time computable) characterisations under $\ont^*$.
\end{itemize}
\end{lemma}
\begin{proof}
We introduce a fresh concept name $A_{\exists S}$ for every basic concept $\exists S$ over $\mathsf{sig}(\ont)$. Let $\ont'$ be the result of substituting $A_{\exists S}$ for every occurrence of the basic concept $\exists S$ in $\ont$ and observe that $\ont'$ is now in named form. Further, let $\ont^*$ be the union of $\ont'$ with
\[\{\exists S\sqsubseteq A_{\exists S}, A_{\exists S}\sqsubseteq\exists S \mid S\in\mathsf{sig}(\ont)\;\text{or}\;\overline{S}\in\mathsf{sig}(\ont)\}\] 
Clearly $\ont^*$ is in named form and can be computed from $\ont$ in polynomial time. Moreover, since $\ont^*$ is a conservative extension of $\ont$ \cite{DBLP:conf/rweb/Krotzsch12}, it is clear that points 1 and 2 hold. Hence, it rests to show point 3.

Let $\mathcal{L}[\Sigma_{\mathsf{C}},\Sigma_{\mathsf{R}}]$ be a concept language over the finite signature $(\Sigma_{\mathsf{C}},\Sigma_{\mathsf{R}})$.
Let $\Sigma_{\mathsf{C}}'$ be the union of $\Sigma_{\mathsf{C}}$ with $\mathsf{sig}(\ont)\cap\NC$ and similarly let $\Sigma_{\mathsf{R}}'$ be the union of $\Sigma_{\mathsf{R}}$ with $\mathsf{sig}(\ont)\cap\NR$. Further, let 
$\Sigma_{\mathsf{C}}''$ be the union of $\Sigma_{\mathsf{C}}$ with $\mathsf{sig}(\ont^*)\cap\NC$ and similarly let $\Sigma_{\mathsf{R}}''$ be the union of $\Sigma_{\mathsf{R}}$ with $\mathsf{sig}(\ont^*)\cap\NR$. Observe that as $\mathsf{sig}(\ont)\subseteq\mathsf{sig}(\ont^*)$ also $\Sigma_{\mathsf{C}}'\subseteq\Sigma_{\mathsf{C}}''$ and $\Sigma_{\mathsf{R}}'\subseteq\Sigma_{\mathsf{R}}''$.

We show that every pointed interpretation $(\Imc,d)$ for the signature $(\Sigma_{\mathsf{C}}',\Sigma_{\mathsf{R}}')$ such that $\Imc\models\ont$ can be transformed in polynomial time to an interpretation $\Imc^*$ for the signature $(\Sigma_{\mathsf{C}}'',\Sigma_{\mathsf{R}}'')$ such that $\Imc^*\models\ont^*$ and for all concepts $D\in\mathcal{L}[\Sigma_C,\Sigma_R]$, $d\in D^{\Imc}$ iff $d\in D^{\Imc^*}$. Given such $(\Imc,d)$, define $\Imc^*$ by keeping the domain and the interpretations of all the old symbols (i.e. those in $\Sigma_{\mathsf{C}}''\setminus\Sigma_{\mathsf{C}}'$ and $\Sigma_{\mathsf{R}}''\setminus\Sigma_{\mathsf{C}}'$) fixed, while setting
\[A_{\exists S}^{\Imc^+}:=\{d\in\Delta^{\Imc} \mid d\in(\exists S)^{\Imc}\}\]
In particular, this means that the interpretation of all the symbols in $\Sigma_{\mathsf{C}}$ and $\Sigma_{\mathsf{R}}$ will remain fixed, from which it follows that $d\in D^{\Imc}$ iff $d\in D^{\Imc^*}$ for all $D\in\mathcal{L}[\Sigma_{\mathsf{C}},\Sigma_{\mathsf{R}}]$. Moreover, clearly this construction is polynomial and by definition of $\Imc^*$ we have $A_{\exists S}^{\Imc^*}=(\exists S)^{\Imc^*}$ for all new symbols $A_{\exists S}\in\mathsf{sig}(\ont^*)\setminus\mathsf{sig}(\ont)$, from which it follows easily that $\Imc^*\models\ont^*$.

The mapping in the other direction just forgets about the interpretation of all the symbols in $\mathsf{sig}(\ont^*)\setminus\mathsf{sig}(\ont)$ (note that $\mathsf{sig}(\ont^*)\setminus\mathsf{sig}(\ont)=(\Sigma_{\mathsf{C}}''\setminus\Sigma_{\mathsf{C}}')\cup(\Sigma_{\mathsf{R}}''\setminus\Sigma_{\mathsf{R}})$). Since for every interpretation $\Jmc$ for $(\Sigma_{\mathsf{C}}'',\Sigma_{\mathsf{R}}'')$ that satisfies $\ont^*$ we have $A_{\exists S}^{\Jmc}=(\exists S)^{\Jmc}$, it follows by the same reasoning as above that its reduct to the smaller signature $(\Sigma_{\mathsf{C}}',\Sigma_{\mathsf{R}}')$ satisfies $\ont$ and also that the element $d$ satisfies the same $\mathcal{L}[\Sigma_{\mathsf{C}},\Sigma_{\mathsf{R}}]$ concepts in this reduct as in $\Jmc$.

It follows that every finite characterisation of $C$ w.r.t.~$\mathcal{L}[\Sigma_{\mathsf{C}},\Sigma_{\mathsf{R}}]$ under $\ont$ can be transformed in polynomial time to a finite characterisation of $C$ w.r.t.~$\mathcal{L}[\Sigma_{\mathsf{C}},\Sigma_{\mathsf{R}}]$ under $\ont^*$ and vice versa. Hence, we conclude that $\mathcal{L}[\Sigma_{\mathsf{C}},\Sigma_{\mathsf{R}}]$ admits finite or polynomial time computable characterisations under $\ont$ iff $\mathcal{L}[\Sigma_{\mathsf{C}},\Sigma_{\mathsf{R}}]$ admits finite or polynomial time computable characterisations under $\ont^*$.
 \end{proof}
\section{Proofs of Section~\ref{sec:prelims}}\label{app:sectwoproofs}

Observe that $\lang{\forall,\exists,\geq,\sqcap,\sqcup,\top,\bot,\neg}$ is a notational variant of $\mathcal{ALCQI}$. We prove that $\mathcal{ALCQI}$ is finitely controllable w.r.t. DL-Lite ontologies (Thm.~\ref{thm:finitecontrollability}). To this end, we quickly review the notion of $\mathcal{ALCQI}$-bisimulations. For technical reasons we immediately introduce a parameterised version.

\begin{definition}[$\mathcal{ALCQI}$-bisimulations]
Let $\Imc,\Imc'$ be two interpretations for the signature $(\Sigma_{\mathsf{C}},\Sigma_{\mathsf{R}})$ and let $k$ and $m$ be natural numbers greater than $0$. Let $Z_1\subseteq\ldots\subseteq Z_m\subseteq\Delta^{\Imc}\times\Delta^{\Imc'}$ be a set of binary relations on the domains of $\Imc$ and $\Imc'$. 
We say that the tuple $(Z_1,\ldots,Z_m)$ of non-empty sets is an \emph{$\mathcal{ALCQI}$ $m$-bisimulation up to $k$} (for the signature $(\Sigma_{\mathsf{C}},\Sigma_{\mathsf{R}})$) between $\Imc$ and $\Imc'$ if for every $1\leq i\leq m$ and $(d,d')\in Z_i$, the following conditions hold:
\begin{itemize}
    \item[(atom)] $d\in A^{\Imc}$ iff $d'\in A^{\Imc'}$ for all $A\in\Sigma_{\mathsf{C}}$
    \item[(forth)] If $i<m$, then for every $S\in\Sigma_{\mathsf{R}}\cup\{R^- \mid R\in\Sigma_{\mathsf{R}}\}$ and $e_1,\ldots e_{k'}\in S^{\Imc}[d]$ pairwise-distinct with $k'\leq k$, there are $e'_1,\ldots e'_{k'}\in S^{\Imc'}[d']$ pairwise-distinct with $(e_j,e'_j)\in Z_{i+1}$ for every $1\leq j\leq k'$
    \item[(back)] symmetric to the (forth)
\end{itemize}
If $(Z_1,\ldots,Z_m)$ is an $m$-bisimulation up to $k$ between $\Imc$ and $\Imc'$ with $(d,d')\in Z_1$, we say that 
there is a $m$-bisimulation up to $k$ linking $d$ to $d'$ and write $(\Imc,d)\;\underline{\leftrightarrow}^{k}_{m}\;(\Imc',d')$.
\end{definition}

\revision{We say that a concept $C$ is \emph{preserved} under $m$-bisimulations (up to $k$ if whenever $(\Imc,d)\;\underline{\leftrightarrow}^{k}_{m}\;(\Imc',d')$ and $d\in C^{\Imc}$ then also $d'\in C^{\Imc'}$.}

\begin{proposition}\label{prop:ALCQIpreservation}
If $C\in\mathcal{ALCQI}[\Sigma_{\mathsf{C}},\Sigma_{\mathsf{R}}]$ then $C$ is preserved under $\mathcal{ALCQI}$ $\mathsf{dp}(C)$-bisimulations up to $\mathsf{nr}(C)$.
\end{proposition}
\begin{proof}
An easy consequence of the above clauses.
\end{proof}

\finitecontrollability*

\begin{proof}
Since the language is closed under all the Boolean operations, it suffices to show that every DL-Lite ontology $\ont$ and every $\mathcal{ALCQI}$ concept $C$ over the finite signature 
$\Sigma_{\mathsf{C}},\Sigma_{\mathsf{R}}$ that is satisfiable w.r.t. $\ont$ is already satisfiable in in a finite interpretation satisfying $\ont$. So let $\ont$ be a DL-Lite ontology and $C$ be a concept from the language $\mathcal{ALCQI}[\Sigma_{\mathsf{C}},\Sigma_{\mathsf{R}}]$ (for some finite $\Sigma_{\mathsf{C}}\subseteq\NC,\Sigma_{\mathsf{R}}\subseteq\NR$) such that $C$ is satisfiable w.r.t. $\ont$. Let $(\Imc,d)$ be a pointed interpretation satisfying $C$ such that $\Imc\models\ont$. If $|\Imc|$ is finite, we are done, so assume $|\Imc|$ is infinite with enumeration $\Delta^{\Imc}=\{d_1,d_2,d_3,\ldots\}$ where $d=d_1$. Via a standard modal unravelling argument, we may assume (up to bisimulation) that $\Imc$ is a tree, meaning that for every $e\in\Delta^{\Imc}$ there is a unique path
\[d=e_1\xrightarrow{S_1}\ldots\xrightarrow{S_n}e_{n+1}=e\]
in $\Imc$, where for each $1\leq i\leq n$ either $S_i\in\Sigma_{\mathsf{R}}$ or $\overline{S}_i\in\Sigma_{\mathsf{R}}$. Then the distance $d_{\Imc}(d,e)$ between $d$ and $e$ in $\Imc$ is $n$. For each node $e\in\Delta^{\Imc}$ at distance $k\leq\mathsf{dp}(C)$ from $d$ in $\Imc$, we define $tp(e):=\{D\in\mathcal{ALCQI}[\Sigma_{\mathsf{C}},\Sigma_{\mathsf{R}}]\;|\;\mathsf{dp}(D)\leq \mathsf{dp}(C)-k,\;\mathsf{nr}(D)\leq\mathsf{nr}(C), \;e\in D^{\Imc}\}\;\cup\;\{B\;\text{basic concept}\;|\;B\;\text{occurs in}\;\ont,\; e\in B^{\Imc}\}$. There are only finitely many such types and hence there are only finitely many types realized in $\Imc$. Note that the type depends on the height of the node (i.e. its distance from the root $d$); if the node $e$ is of height $k$ only formulas of depth at most $\mathsf{dp}(C)-k$ occur in $tp(e)$.\footnote{Note that the notion of height of a node is defined differently then the height $h_R(\Imc,d)$ of a pointed interpretation $(\Imc,d)$ w.r.t. a role name $R$.}

For each $0\leq i\leq \mathsf{dp}(C)$, define a relation $Z_i\subseteq\Delta^{\Imc}\times\Delta^{\Imc}$ as follows.
\[Z_i:=\{(e,e')\in\Delta^{\Imc}\times\Delta^{\Imc}\;|\;d_{\Imc}(d,e)=d_{\Imc}(d,e')=i,\]
\[tp(e)=tp(e')\;\text{and}\;e'\;\text{is among the first}\;\mathsf{dp}(C)\;\text{elements}\]
\[\text{in the enumeration of}\;\Delta^{\Imc}\;\text{of this height and type}\}\]
In particular, this implies that $Z_0=\{(d,d)\}$. Let $Z:=\bigcup_{0\leq i\leq \mathsf{dp}(C)}$ and let $\Imc'$ be the substructure induced by the set $cod(Z)=\{e\;|\;\exists e'\;(e',e)\in Z\}$, which is a subset of $\Delta^{\Imc}$. Clearly $|\Imc'|$ is finite since for each of the finitely many types at most $\mathsf{nr}(C)$ elements of $\Delta^{\Imc}$ realizing that type can occur in the codomain of $Z$. We claim that $(Z_0,\ldots,Z_{\mathsf{dp}(C)})$ is a $\mathsf{dp}(C)$-bisimulation up to $\mathsf{nr}(C)$ linking $d$ in $\Imc$ to $d$ in $\Imc'$, from which it follows that $d\in C^{\Imc'}$ by Prop.~\ref{prop:ALCQIpreservation}.

Since each $Z_i$ only links elements in $\Delta^{\Imc}$ that satisfy the same atomic (depth $0$) type, it follows that the (atom) clause holds for all $Z_i$. Note that thereby the relation $Z_m$ satisfies all the required clauses. Let $0\leq i<m$. Clearly the (back) clause holds for all pairs in $Z_i$ since for each $e'\in cod(Z_i)$, we have that $(e',e')\in Z_i$ and $\Imc'\subseteq\Imc$. For the (forth) clause, consider a pair $(e,e')\in Z_i$, let $s$ be a role name or the inverse of a role name in $\Sigma_{\mathsf{R}}$ and let $t_1,\ldots, t_k\in S^{\Imc}[e]$ be distinct $s$-successors of $e$, where $k\leq\mathsf{nr}(C)$. It follows that $\{t_1,\ldots,t_k\}$ contains at most $k\leq\mathsf{nr}(C)$ realizations of a type, and hence there are sufficiently many distinct $s$-successors of $e'$ in $\Imc$ that realize those same types, and are of the same height and therefore $Z_i$-related to the appropriate $t_i$'s. This shows that each $Z_i$ satisfies the (forth) clauses and hence $(Z_0,\ldots, Z_{\mathsf{dp}(C)})$ is a $\mathsf{dp}(C)$-bisimulation up to $\mathsf{nr}(C)$ linking $d$ in $\Imc$ to $d$ in $\Imc'\subseteq\Imc$.

Define an extension $\Jmc$ of $\Imc'$ as follows.
\begin{itemize}
    \item $\Delta^{\mathcal{J}}:=\Delta^{\Imc'}\cup\{d_{\exists S} \mid S\in\mathsf{sig}(\ont)\;\text{or}\;\overline{S}\in\mathsf{sig}(\ont),\exists S\;\text{is satisfiable w.r.t.}\;\ont\}$
    \item For all $A\in\Sigma_{\mathsf{C}}'$, $A^{\mathcal{J}}:=A^{\Imc'}\cup\{d_{\exists S} \mid \ont\models\exists S\sqsubseteq A\}$
    \item For all $R\in\Sigma_{\mathsf{R}}'$, $R^{\mathcal{J}}$ is the union of $R^{\Imc'}$ with $\{(e,d_{\exists S^-})\in\Delta^{\Imc'}\times\Delta^{\mathcal{J}} \mid \ont\models B\sqsubseteq \exists S, e\in (B\sqcap\neg\exists S)^{\Imc'},d(d,e)=\mathsf{dp}(C)\}\cup $ \newline
    $\{(d_{\exists S},d_{\exists T^-})\in\Delta^{\mathcal{J}}\times\Delta^{\mathcal{J}} \mid \ont\models \exists S\sqsubseteq\exists T\}\cup $ \newline
    $\{(d_{\exists T},d_{\exists S})\in\Delta^{\mathcal{J}}\times\Delta^{\mathcal{J}} \mid \ont\models\exists S\sqsubseteq \exists T^-\}$
\end{itemize}
That is, the only edges between elements in $\Delta^{\Imc'}$ and $\Delta^{\Jmc}\setminus\Delta^{\Imc}$ go between leaves of the tree $\Imc'$ (at distance $\mathsf{dp}(C)$ from the root $d$) and newly added elements of the form $d_{\exists S}$. It follows that $(\Imc',d)\;\underline{\leftrightarrow}^{\mathsf{nr}(C)}_{\mathsf{dp}(C)}\;(\mathcal{J},d)$ because the neighborhood of $d$ in $\Imc'$ and $\mathcal{J}$ at distance $\mathsf{dp}(C)$ remains the same. Thus, by Prop.~\ref{prop:ALCQIpreservation} $d\in C^{\mathcal{J}}$ as well. It rests to show that $\Jmc\models\ont$, so let $B\sqsubseteq C\in\ont$.

Suppose $e=d_{\exists S}$ is an element of $\Delta^{\Jmc}\setminus\Delta^{\Imc}$ and that $e\in B^{\Jmc}$. It follows straight from the definition of $\Jmc$ that for every basic concept $B$ over $\mathsf{sig}(\ont)$ we have $d_{\exists S}\in B^{\Jmc}$ iff $\ont\models\exists S\sqsubseteq B$. Hence we know that $d_{\exists S}\in B^{\Jmc}$ implies that $\ont\models\exists S\sqsubseteq B$, which together with $\ont\models B\sqsubseteq C$ implies that $\ont\models\exists S\sqsubseteq C$. Using the equivalence again we get that $d_{\exists S}\in C^{\Jmc}$.

Now suppose that $e\in\Delta^{\Jmc}$ at distance $k<\mathsf{dp}(C)$ from $d$ with $e\in B^{\Jmc}$. In this case it is clear that $(\Imc,e)\;\underline{\leftrightarrow}^{\mathsf{nr}(C)}_{1}\;(\Imc',e)\;\underline{\leftrightarrow}^{\mathsf{nr}(C)}_{1}\;(\Jmc,e)$.
Since $\mathsf{dp}(C),\mathsf{dp}(D)\leq 1$ and $e\in B^{\Jmc}$, via the bisimulation $e\in B^{\Imc}$, but as $\Imc\models\ont$ we get $e\in C^{\Imc}$ from which we conclude that $e\in C^{\Jmc}$ via the bisimulation again.

Finally, suppose that $e\in\Delta^{\Imc'}$ is at distance $\mathsf{dp}(C)$ from $d$ with $e\in B^{\Jmc}$. It follows from the definition of $\Jmc$ that $e\in B^{\Jmc}$ only if there is some basic concept $B'$ such that $\ont\models B'\sqsubseteq B$ and $e\in (B'\sqcap\neg B)^{\Imc'}$. Note that by concatenating these subsumptions we get $\ont\models B'\sqsubseteq C$. If $C$ is a concept name or the negation of one, it follows that $e\in C^{\Jmc}$ because $(\Imc,e)\;\underline{\leftrightarrow}^{0}_{0}\;(\Jmc,e)$ and $\Imc\models\ont$.

Next, suppose that $C$ is of the form $\exists S$. If $e\in(\exists S)^{\Imc'}$, it follows that $e\in(\exists S)^{\Jmc}$ as well as $\Imc'\subseteq\Jmc$. Else $e\in (B'\sqcap\neg\exists S)^{\Jmc}$ so by definition of $\Jmc$, $(e,d_{\exists S^-})\in R^{\Jmc}$ and thus $e\in(\exists S)^{\Jmc}$. Finally, suppose that $C$ is of the form $\neg\exists S$. Since $\Imc\models\ont$ and $\ont\models B'\sqsubseteq\neg\exists S$, we get that $e\in(\neg\exists S)^{\Imc}$. But as $\Imc'\subseteq\Imc$ it follows that $e\in(\neg\exists S)^{\Imc'}$. It can only be that an $S$-successor is added to $e$ in $\Jmc$ if $\ont\models B''\sqsubseteq\exists S$ for some basic concept $B''$ with $e\in(B''\sqcap\neg\exists S)^{\Imc}$. But in that case $e\in(B'')^{\Imc}$ and since $\Imc\models\ont$ we get $e\in(\exists S\sqcap\neg\exists S)^{\Imc}$; a contradiction. Hence it cannot be that an $S$-successor is added to $e$ in $\Jmc$ and hence $e\in(\neg\exists S)^{\Jmc}$.
\end{proof}

\section{Proofs of Section~\ref{sec:withoutontologies}}

We start by proving Thm.~\ref{thm:ALEQcharacterisations}. To this end, we first define a set of interpretations that characterises $\top$ w.r.t. $\lang{\forall,\exists,\geq,\sqcap,\top}$.

Let $\Imc$ be the interpretation consisting of a single point with no successors, i.e. $\Delta^{\Imc}:=\{d\}$ and $R^{\Imc}:=\emptyset$ and $A^{\Imc}:=\emptyset$ for all $R\in\Sigma_{\mathsf{R}}$ and $A\in\Sigma_{\mathsf{C}}$. Further, let $\Imc'$ be the interpretation with one loopstate with a deadlock successor, i.e. $\Delta^{\mathcal{I'}}:=\{d',e'\}$ and $R^{\mathcal{I'}}:=\{(d',d'),(d',e')\}$ and $A^{\mathcal{I'}}:=\emptyset$ for all $R\in\Sigma_{\mathsf{R}}$ and $A\in\Sigma_{\mathsf{C}}$.

\begin{lemma}
$E^+_{\top}:=\{(\Imc,d),(\mathcal{I'},d')\}$ is a finite characterisation of $\top$ w.r.t. $\lang{\forall,\exists,\geq ,\sqcap,\top}$.
\end{lemma}
\begin{proof}
Let $C\in\lang{\forall,\exists,\geq,\sqcap,\top}$ be arbitrary. We show that $C$ fits $E^+_{\top}$ iff $C\equiv\top$ by induction on the role depth of $C$. For the base case, if a depth $0$ concept fits these examples, clearly it must be a conjunction of $\top$'s and hence equivalent to $\top$, because no concept name is made true at point $d$ in $\Imc$ (or on any other positive example). 

For the induction step, let $C$ be of role depth $n+1$. Since $d\in C^{\Imc}$ we know that $C$ is equivalent to a concept of the form $\forall R_{1}.C_1\sqcap\ldots\forall R_m.C_m$ where $R_1,\ldots,R_m\in\Sigma_{\mathsf{R}}$ and $C_1,\ldots,C_m$ have role depth at most $n$. For each $1\leq i\leq m$, since $d'\in(\forall R_i.C_i)^{\mathcal{I'}}$ it follows that both $d',e'\in C_i^{\mathcal{I'}}$. 

However, it is easily seen that $(\mathcal{I'},e')$ is counting bisimilar to $(\Imc,d)$, where the relation defining the counting bisimulation is simply the singleton pair $\{(d,e')\}$ (this is a counting bisimulation because both nodes have no successors). Hence by preservation of $\lang{\forall,\exists,\geq,\sqcap,\top,\bot}$ under counting bisimulations, $d\in C_i^\Imc$. But then $C_i$ fits $E^+_{\top}$ so by the inductive hypothesis $C_i\equiv\top$. Since this holds for all $1\leq i\leq m$ and $\forall R.\top\equiv\top$ for all $R\in\Sigma_{\mathsf{R}}$, it follows that $C\equiv\top\sqcap\ldots\sqcap\top\equiv\top$.
\end{proof}

Now we are in a position to prove Prop.~\ref{prop:knrdepthn}.
As a warming up exercise, we first show an easy special case.

\begin{proposition}
For $n\geq 0$, there is an interpretation $(\Imc_n,d)$ (which can be constructed in polynomial time), such that
    for all $\lang{\exists,\sqcap,\top}$-concepts $C$,    
    the following are equivalent:
    \begin{enumerate}
    \item $C$ is equivalent to a concept expression of role depth at most $n$;
    \item $C$ fits the positive example $(\Imc_n,d)$.    \end{enumerate}
\end{proposition}

\begin{proof}
Let $\Imc_n$ be the interpretation
given by 
    $\Delta^{\Imc_n} = 
    A^{\Imc_n}=\{d_0, \ldots, d_n\}$ for all $A\in\Sigma_{\mathsf{C}}$, 
    and $R^{\Imc_n}=\{(d_i,d_{i+1})\mid i< n\}$
    for all $R\in\Sigma_{\mathsf{R}}$. That is, $\Imc_n$ is the 
    directed path
        \[ d_0\xrightarrow{\bigcap_{R\in\NR}R} d_1 \xrightarrow{\bigcap_{R\in\NR}R}  \cdots d_{n-1} \xrightarrow{\bigcap_{R\in\NR}R} d_n\]
    where each element satisfies all $A\in\Sigma_{\mathsf{C}}$.
    An easy induction argument shows
    that, for all natural numbers $k\leq n$ and for all 
    $\lang{\exists,\sqcap,\top}$-concept expressions $C$,
    $C$ has role depth $k$ if and only if $d_{n-k}\in C^{\Imc_n}$. In particular, this shows that the proposition holds (take $k=n$ and $d=d_0$).
    \end{proof}

We can do something similar for 
$\lang{\forall,\exists,\geq,\sqcap,\top}$:

\knrdepthn*

\begin{proof}
Given a pointed interpretation $(\Imc,d)$, we define a new pointed interpretation $(\Imc,d)^n_k,d_1$:
\begin{itemize}
\item  $\Delta^{(\Imc,d)^n_k}:=\Delta^{\Imc}\cup\{d_1,d^{1}_{2},\ldots,d^{k}_{2},\ldots,d^{1}_{n},\ldots,d^{k}_{n}\}$, 
\item $r^{(\Imc,d)^n_k}:=r^{\Imc}\cup\{(d_1,d^j_2)\;|\;1\leq j\leq k\}\cup\{(d^j_i,d^{j'}_{i+1})\;|\;2\leq i<n, 1\leq j,j'\leq k\}\cup\{(d^j_n,d)\:|\;1\leq j\leq k\}$ for all $r\in\NR$ and
\item $A^{(\Imc,d)^n_k}:=A^{\Imc}\cup\{d_1,d^{1}_{2},\ldots,d^{k}_{2},\ldots,d^{1}_{n},\ldots,d^{k}_{n}\}$ for all $A\in\NC$.     
\end{itemize}
Finally, set $D_{n,k}:=\{(\Imc,d)^n_k,d_1)\;|\;(\Imc,d)\in E^+_{\top}\}$. In other words, $D_{n,k}$ consists of all pointed interpretations $(\Imc,d)^n_k$ with distinguished point $d_1$ of the form:
        \[ 
        d_1
        \xrightarrow{\bigcap_{R\in\NR}R} 
        \begin{array}{c}d_2^1\\\vdots\\d_2^k\end{array}
        \xrightarrow{\bigcap_{R\in\NR}R}
        \cdots 
        \begin{array}{c}d_n^1\\\vdots\\d_n^k\end{array}
        \xrightarrow{\bigcap_{R\in\NR}R} 
        \fbox{$(\mathcal{I},d)$}\]

with $(\mathcal{I},d)\in E^+_\top$ (with edges running from $d_1$ to each $d_2^i$, from each $d_i^j$ to each $d_{i+1}^{j'}$, and from each $d_n^i$ to $d'$).

Suppose that some $C\in\lang{\forall,\exists,\geq,\sqcap,\top}$ fits $D_{n,k}$ as positive examples.
By the same argument as before, $\mathsf{dp}(C)$ must be bounded by $n$. Moreover, since the width (i.e. the degree as a graph) of the examples in $D_{n,k}$ is bounded by $k$, it also follows that $\mathsf{nr}(C)\leq k$. To see this, first observe that every domain element in $(\Imc,d)^n_k$  and $(\Imc',d')^n_k$ has \emph{at most} $k$ successors. Thus, if there is some concept expression $C\in\lang{\forall,\exists,\geq,\sqcap,\top}$ with some sub-expression $C'$ which has a conjunct of the form $\geq k' C''$ where $k'>k$ that fits $D_{n,k}$, then $C''$ must be true at some domain element in  $(\Imc_0,d)^n_k$ or $(\Imc',d')^n_k$, which cannot be as each of these elements has at most $k<k'$ successors.
 
Conversely, it is not difficult to see that any concept $D\in\lang{\forall,\exists,\geq,\sqcap,\top}$ with $\mathsf{dp}(D)\leq n$ and $\mathsf{nr}(C)\leq k$ is satisfied on all the examples in $D_{n,k}$, because every point $d_i^j$ satisfy all basic concepts by construction and the interpretations are ``fat'' in the sense that every point $d_i^j$ has $k$ counting-bisimilar successors.
\end{proof}

Finally, Thm.~\ref{thm:ALEQcharacterisations} follows from Prop.~\ref{prop:knrdepthn}.

\ALEQcharacterisations*

\begin{proof}
Every concept expression $C$ in $\lang{\forall,\exists,\geq,\sqcap,\top}$ has a finite role depth $n$ and maximum number restriction $k$. Hence by Prop.~\ref{prop:knrdepthn} 
we know that $C$ fits $E^+_{n,k}$. Moreover, since there are only finitely many concepts expression up to logical equivalence that have maximal role depth $n$ and maximal number restriction $k$, it follows that there is some natural number $m$ and concept expressions $C_1,\ldots, C_m$ such that every concept expression that fits $E^+_{n,k}$ is equivalent to $C_i$ for some $1\leq i\leq m$. In  particular, there is some $C_i$ which is equivalent to $C$. But then by Thm.~\ref{thm:finitecontrollability} we can pick $m-1$ examples $(\Imc_1,d_1),\ldots,(\Imc_{m-1},d_{m-1})$ distinguishing $C_i$ from all $C_j$ with $i\ne j$. Then clearly adding these (positive or negative) examples to $E^+_{n,k}$ results in a finite characterisation of $C$ w.r.t. $\lang{\forall,\exists,\geq,\sqcap,\top}$.
\end{proof}

We proceed with the proofs of the results leading up to Thm.~\ref{thm:ELQcharacterisations}, which provides an elementary (doubly exponential) construction of finite characterisations for the fragment $\lang{\exists,\geq,\sqcap,\top}$ of $\lang{\forall,\exists,\geq,\sqcap,\top}$.

\begin{restatable}{lemma}{ELQsubsumption}
\label{lem:ELQsubsumption}
For all $\lang{\geq,\sqcap,\top}$-concepts 
\[ C = A_1\sqcap\cdots\sqcap A_m \sqcap \underset{1\leq i\leq n}{\bigsqcap}\geq k_i\;R_i.C_i\]
and
\[ C' = A'_i\sqcap\cdots\sqcap A'_{m'}\sqcap \underset{1\leq i\leq n'}{\bigsqcap}\geq k'_i\;R'_i.C'_i \]
(where $m,n,m',n'\geq 0$) the following are equivalent:
\begin{enumerate}
    \item $\models C\sqsubseteq C'$
    \item $\{A'_1, \ldots, A'_{m'}\}\subseteq \{A_1, \ldots, A_m\}$ and for every $j\leq n'$ there is an $i\leq n$ such that
    $k_i\geq k'_j$, $R_i=R'_j$, and 
    $\models C_i\sqsubseteq C'_j$.
\end{enumerate}
\end{restatable}

\begin{proof}
We will prove the direction from 1 to 2. The other direction is immediate. Assume that $\models C\sqsubseteq C'$. It is easy to construct a pointed interpretation $(\Imc,d)$ satisfying $C$ in which $d$ satisfies the
concept names $A_1, \ldots, A_m$ and no other concept name. Since $\models C\sqsubseteq C'$ it must be that $d$ satisfies all of $A'_1,\ldots, A'_{m'}$ as well and hence $\{A'_1,\ldots,A'_{m'}\}\subseteq\{A_1,\ldots,A_m\}$. Next, fix any $j\leq n'$, and consider
the conjunct $\geq k'_j R'_j.C'_j$.
To simplify the notation in what follows,
let $k=k'_j$, $R=R'_j$, and $D=C'_j$.
We must show that there exists an $i\leq n$
such that $k_i\geq k$, $R_i=R$, and $\models C_i\sqsubseteq D$. For the purpose of deriving a contradiction, suppose otherwise.

By assumption, for each $C_i$ with $R_i=R$ and $k_i\geq k$ there is some pointed interpretation $(\Imc_i,d_i)$ such that $d_i\in(C_i\sqcap\neg D)^{\Imc_i}$. Further, let $(\Jmc,d)$ be a pointed interpretation satisfying the concept
\[C^{-R}:=A_1\sqcap\cdots\sqcap A_m \sqcap\underset{1\leq i\leq n, R_i\ne R}{\bigsqcap}\geq k_i R_i.C_i\]
(that is, the concept obtained from $C$ by removing the conjuncts that have a number restriction with the role $R$ as their main operand). We can furthermore ensure (through unraveling and taking a suitable sub-interpretation) that $d$ has no $R$-successors in $\mathcal{J}$. Finally, let $(\Jmc',d')$ be any model of the conjunction of all concepts in $C_i$ for which $R_i=R$ and $k_i< k$. We will now create a pointed interpretation $(\Imc,d)$ such that $e\in(C\sqcap\neg\geq k R.D)^{\Imc}$.

First, let $\Imc'$ be the disjoint union of $\mathcal{J}$ together with $k-1$ isomorphic copies of $\Jmc'$ and $k_i$ isomorphic copies of $\Imc_i$ for each $C_i$ with $R_i=R$ and $k_i\geq k$. We denote by $d'^l$ the $l$-th copy of the designated element $d'\in\Delta^{\mathcal{J'}}$ (where $1\leq l\leq k-1$) and by $d_i^j$ the $j$-th copy of the designated element $d_i\in\Delta^{\Imc_i}$ (where $1\leq j\leq k_i$ and $1\leq i\leq n$). Finally, $\Imc$ is obtained from $\Imc'$ by  setting $R^{\Imc}:=R^{\Imc'}\cup\{(d,d'^l)\;|\;1\leq l\leq k-1\}\cup\{(d,d_i^j)\;|\;1\leq j\leq k_i\}$ (where $d$ is the designated element of $\mathcal{J}$). In a picture, $(\Imc,d)$ looks as follows (where $\{C_{i_1},\ldots C_{i_l}\}$ is the set of all $C_i$ with $R_i=R$ and $k_i\geq k$).

\vspace{0.1in}

\begin{tikzpicture}[sibling distance=0.45cm,level 2/.style={sibling distance =0.3cm}]

\node[] {$d$}
        child{ node[rtria] {$\mathcal{J}$} 
        edge from parent [draw=none] }
        child [missing] {}
        child [missing] {}
        child [missing] {} 
        child{ node[itria] {$\mathcal{J'}$} 
            child {node[fill=white] {{\footnotesize $k-1$ copies}}
            edge from parent [dashed] }
            }
        child [missing] {}
        child [missing] {}
        child{ node[fill=white] {$\ldots$}
        edge from parent [draw=none] }
        child{ node[itria] {$\Imc_1$} 
            child {node[fill=white] {{\footnotesize $k_{i_1}$ copies}}
            edge from parent [dashed] }
            }
        child [missing] {}
        child{ node[fill=white] {$\ldots$}
        edge from parent [draw=none] }
        child [missing] {}
        child{ node[itria] {$\Imc_n$} 
            child {node[fill=white] {{\footnotesize $k_{i_l}$ copies}}
            edge from parent [dashed] }
            };
\end{tikzpicture}
To reach our contradiction, it suffices to prove:

\medskip\par\noindent\emph{Claim: $d\in(C\sqcap\neg\geq k R.D)^{\Imc}$.}

\medskip\par\noindent 
First of all, note that $d_i^j\in(C_i\sqcap\neg D)^{\Imc}$ for all $1\leq j\leq k_i$ and $d'^l\in C_i^{\Imc}$ for all $1\leq l\leq k-1$ and all $C_i$ with $R_i=R$ and $k_i<k$. This is because $\lang{\exists,\geq,\sqcap,\top}$ concepts are preserved under generated submodels, and the submodel generated by the element $d_i^j\in\Delta^{\Imc}$ is exactly $\Imc_i^j$ (i.e. the $j$'th isomorphic copy of $\Imc_i$ in $\Imc$), and likewise the submodel generated by the element $d'^l\in\Delta^{\Imc}$ is exactly $\mathcal{J'}^l$ (i.e. the $l$'th copy of $\mathcal{J'}$ in $\Imc$). It follows that $d\in C^{\Imc}$. Thus it remains to show that $d\not\in(\geq k R.D)^{\Imc}$.

By construction, we have that $R^{\Imc}[d]=\{d'^l \mid 1\leq l\leq k-1\}\cup\{d_i^j \mid 1\leq j\leq k_i, C_i\in\mathcal{C}_{R}^{\geq k}\}$ since we assumed that $R^{\mathcal{J}}[d]=\emptyset$. But then the claim follows since $d_i^j\not\in D^{\Imc}$ for all $1\leq i\leq n$ and $1\leq j\leq k_i$ and $d$ has only $k-1$ other $R$-successors (namely $d'_1,\ldots,d'_{k-1}$).
\end{proof}

\begin{remark} Note that a straightforward generalisation of Lem.~\ref{lem:ELQsubsumption} to its extension $\lang{\exists,\geq,-,\sqcap,\top}$ with inverse roles does not hold, as witnessed by the following example. Let $C:=\exists R.A\;\sqcap\;\exists R.B$ and let $D:=\exists R.(A\;\sqcap \;B)$ and $D':=\;\geq 2 R.\top$. Note that $\not\models C\sqsubseteq D$ and $\not\models C\sqsubseteq D'$ but $\models C\sqsubseteq D\sqcup D'$. Consider the following concepts:
\begin{align*}
    C_0:=&\exists R.(A\sqcap\exists R^-.D)\sqcap\exists R.(B\sqcap\exists R^-.D') \\
    D_0:=&\exists R.(\exists R^-.D\sqcap\exists R^-.D')
\end{align*}
and observe that $\models C_0\sqsubseteq D_0$. This is because $\models C_0\sqsubseteq D\;\sqcup\;D'$, $\models C_0\sqcap D\sqsubseteq D_0$ and $\models C_0\sqcap D'\sqsubseteq D_0$. Thus, $\models C_0\sqsubseteq D_0$ is a valid concept subsumption that
does not allow for the same type of decomposition
as in Lem.~\ref{lem:ELQsubsumption}.
\end{remark}

\polytimesubsumption*

\begin{proof}
We can use a straightforward dynamic programming approach: 
for each subconcept $D$ of $C$ and for 
each subconcept $D'$ of $C'$, 
we test
if $\models D\sqsubseteq D'$. We do this 
in order of increasing role depth, starting
with subconcepts that are concept names.
Note that Lem.~\ref{lem:ELQsubsumption} allows
us to determine subsumption for pairs of 
sub-concepts of role depth $d$ once we
already tested subsumption or pairs of
sub-concepts of role depth strictly less than $d$.
\end{proof}

\propnormalformpolytime*
\begin{proof}
We repeatedly test, for each conjunction
in $C$, all pairwise subsumptions between
conjuncts, using Theorem~\ref{thm:polytimesubsumption}. If a subsumption 
$\models C_i\sqsubseteq C_j$ holds,
where $C_i,C_j$ are distinct conjuncts
of some conjunction, we remove the conjunct $C_j$. Clearly, this process terminates in 
polynomial time and results in an equivalent
irredundant concept expression.
\end{proof}

\smallsubinterpretation*

\begin{proof}
By induction on $\mathsf{dp}(C)$. If $\mathsf{dp}(C)=0$ and $d\in C^{\Imc}$ then already the submodel $\Jmc$ of $\Imc$ induced by the single point $\{d\}$ is such that $d\in C^{\Jmc}$ and $|\Jmc|=1\leq|C|$. Now suppose the claim holds for all $\lang{\geq,\sqcap,\top}$ concepts of depth at most $n$ and consider a concept $C$ from this language with $\mathsf{dp}(C)=n+1$ of the form

\[A_1\sqcap\cdots\sqcap A_m \sqcap \underset{1\leq i\leq n}{\bigsqcap}\geq k_i\;R_i.C_i\]

By Proposition \ref{prop:normalformpolytime} we can assume without loss of generality that $C$ is irredundant. Let $(\Imc,d)$ be a pointed interpretation with $d\in C^{\Imc}$. Following the semantics, for each $1\leq i\leq n$ there must be $k_i$ distinct successors $d_i^1,\ldots,d_i^{k_i}\in R^{\Imc}[d]$ such that $d_i^j\in C_i^{\Imc}$ for all $1\leq j\leq k_i$.

By the inductive hypothesis, for each such $1\leq i\leq n$ and $1\leq j\leq k_i$ there is a $\Imc_i^j\subseteq\Imc$ and a point $d_i^j\in\Delta^{\Imc_{i}^{j}}$ such that $d_i^j\in C_i^{\Imc_i^j}$ and $|\Imc_i^j|\leq |C_i|^{|C_i|}$. Let $\Imc'$ be the sub-interpretation of $\Imc$ induced by the subset $X=\{d\}\cup\underset{1\leq i\leq n}{\bigcup}\;\underset{1\leq j\leq k_i}{\bigcup}\Delta^{\Imc_i^j}$ of $\Delta^{\Imc}$. Clearly $d\in C^{\Imc'}$ and
\begin{align*}
    |\Imc|=|X| & \leq\Sigma_{i=1}^n\Sigma_{j=1}^{k_i}|\Imc_i^j| \leq \Sigma_{i=1}^n|C|\cdot|C_i|^{|C_i|}\\
    & \leq|C|\cdot\Sigma_{i=1}^n|C_i|^{|C|-1}\leq|C|\cdot(\Sigma_{i=1}^n|C_i|)^{|C|-1}\\
    & \leq|C|\cdot|C|^{|C|-1}=|C|^{|C|}.
    \end{align*}
    where we use that $\Sigma_{i=1}^n|C_i|\leq|C|-1$ (and hence $|C_i|\leq |C|-1$ for all $1\leq i\leq n$), $k_1,\ldots, k_n\leq\mathsf{nr}(C)\leq|C|$ and $\Sigma_{i=1}^na^{b}\leq(\Sigma_{i=1}^na)^b$ for all numbers $a,b$.
\end{proof}

\frontier*

\begin{proof}
The frontier construction is given by induction on the role depth of $C$.
Let $C$ be any $\lang{\geq,\sqcap,\top}$-concept of the form
\[C=A_1\sqcap\ldots\sqcap A_m\sqcap\underset{1\leq i\leq n}{\bigsqcap}\geq k_i R_i.C_i\]
where $m,n\geq 0$. Note that the case where $n=0$ (in which case $C$ has role depth 0) corresponds to the base case of the induction. By Prop.~\ref{prop:normalformpolytime} we may further assume that $C$ is irredundant. Let $\mathcal{F}(C)$ consist of
\begin{enumerate}
     \item[(a)]  all concepts obtained from $C$ by dropping some conjuncts $A_j\in\NC$ for $1\leq j\leq m$, and
     \item[(b)] all concepts $D_i$ obtained from $C$ by replacing a conjunct $\geq k_i R_i.C_i$ (with $i\leq n$) by \[\geq (k_i-1) R_i.C_i\sqcap\underset{D\in \mathcal{F}(C_i)}{\bigsqcap}\geq k_i R_i.D\]
 \end{enumerate}
Note that $\mathcal{F}(\top)=\emptyset$ and that $\mathcal{F}(C)=\{\top\}$ for concept names $C$ as well as for concepts $C$ of the form $\geq 1 R.\top$. We first show that our construction is polynomial. More precisely, we will show that for every concept $C\in\lang{\geq,\sqcap,\top}$ that $\Sigma_{C'\in\mathcal{F}(C)}|C'|\leq 5\cdot\mathsf{dp}(C)\cdot|C|^3$ showing that $\Sigma_{C'\in\mathcal{F}(C)}|C'|=\mathcal{O}(|C|^4)$ since $\mathsf{dp}(C)\leq|C|$. If $\mathsf{dp}(C)=0$, i.e. $C=A_1\sqcap\ldots\sqcap A_m$, we have that $\mathsf{nr}(C)=0$ and $\Sigma_{C'\in F(C)}|C'|=\Sigma_{i=1}^m|A_1\sqcap\ldots A_{i-1}\sqcap A_{i+1}\sqcap\ldots\sqcap A_m|\leq \Sigma_{i=1}^m|C|\leq |C|^2$. 

Next, consider some concept $C$ of depth $\mathsf{dp}(C)>0$ of the form $C=A_1\sqcap\cdots\sqcap A_m \sqcap \underset{1\leq i\leq n}{\bigsqcap}\geq k_i\;R_i.C_i$. By the inductive hypothesis, for all subconcepts $C_i$ we have that $\Sigma_{D\in\mathcal{F}(C_i)}|D|\leq 5\cdot\mathsf{dp}(C_i)\cdot|C_i|^3$. Let $\mathcal{F}_{a}(C)$ be the set of all elements of the frontier $\mathcal{F}(C)$ of the form (a), and similarly we write $\mathcal{F}_{b}(C)$ for the set of all frontier elements of the form (b). Note that $\mathcal{F}_{a}(C)$ and $\mathcal{F}_{b}(C)$ partition $\mathcal{F}(C)$ in two, so $\Sigma_{C'\in\mathcal{F}(C)}|C'|=\Sigma_{C'\in\mathcal{F}_{a}(C)}|C'|+\Sigma_{C''\in\mathcal{F}_{b}(C)}|C''|$. Observe that $|C'|\leq|C|$ for all $C'\in\mathcal{F}_a(C)$ and hence $\Sigma_{C'\in\mathcal{F}_{a}(C)}|C'|\leq|C|^2$. Next, for every concept $C''=D_i\in\mathcal{F}_b(C)$ we have
\[|D_i|\leq|C|+\Sigma_{D\in\mathcal{F}(C_i)}(|D|+|C|+2)\]
because $D_i$ contains a conjunct of the form $\geq k_i R_i.D$ for every $D\in\mathcal{F}(C_i)$ and $|\geq k_i R_i.D|\leq 1+|C|+|D|$ because the number $k_i\leq\mathsf{nr}(C)$ can be represented in less than $|C|$ symbols. The extra $+1$ is incurred by also counting the conjunction symbols. It follows that the total size of $\mathcal{F}(C)$ is smaller than or equal to
\begin{align*}
& \leq |C|^2 + \Sigma_{i=1}^n(|C|+\Sigma_{D\in\mathcal{F}(C_i)}(|D|+|C|+2))\\
& \leq |C|^2 +|C|^2+|C|^3+2|C|^2+\Sigma_{i=1}^n\Sigma_{D\in\mathcal{F}(C_i)}|D|\\
& \leq 5\cdot|C|^3+\Sigma_{i=1}^n(5\cdot\mathsf{dp}(C_i)\cdot|C_i|^3) \\
& \leq 5\cdot|C|^3+5(\mathsf{dp}(C)-1)\Sigma_{i=1}^n|C_i|^3\\
& \leq 5\cdot|C|^3+5(\mathsf{dp}(C)-1)(\Sigma_{i=1}^n|C_i|)^3\\
& \leq 5\cdot|C|^3+5|C|^3(\mathsf{dp}(C)-1)\\
& = 5\cdot|C|^3 + 5\cdot\mathsf{dp}(C)\cdot|C|^3-5|C|^3=5\cdot\mathsf{dp}(C)\cdot|C|^3
\end{align*}
where we use that $n\leq |C|, |\mathcal{F}(C_i)|\leq|C_i|\leq|C|$, $a\Sigma_{i=1}^nb_i=\Sigma_{i=1}^n(a\cdot b_i)$ and $\Sigma_{i=1}^n(a^b)\leq(\Sigma_{i=1}^na)^b$ for all numbers $a,b,b_1,\ldots,b_n$.

\vspace{0.1in}

\emph{Claim: $F(C)$ is indeed a frontier of $C$ w.r.t. $\lang{\geq,\sqcap.\top}$.} We will prove that both items (i) and (ii) from the definition of frontiers hold. For (i), let $D\in F(C)$. If $D$ was obtained through case (a) of the above construction (i.e., by dropping a concept name $A_j$ concept), it is clear that $\models C\sqsubseteq D$, and also $\not\models D\sqsubseteq C$ since $\not\models D\sqsubseteq A_j$. If $D$ was obtained through case (b) of the above construction (i.e., by replacing some conjunct $\geq k_i R_i.C_i$  by $\geq k_i-1 R_i.C_i\sqcap\underset{C'\in F(C_i)}{\bigsqcap}\geq k_i R_i.C'$), then it follows from the induction hypothesis that the replacing concept subsumes the original conjunct. To show that $\not\models D\sqsubseteq C$, first observe that by Lem.~\ref{lem:ELQsubsumption} we see that $\not\models\geq k_i-1 R_i.C_i\sqcap\underset{C'\in F(C_i)}{\bigsqcap}\geq k_i R_i.C'\sqsubseteq\geq k_i R_i.C_i$. Moreover, since $C$ is irredundant, $\geq k_i R_i.C_i$ does not subsume any other conjunct of $C$. But then it follows from Lem.~\ref{lem:ELQsubsumption} that $\geq k_i R_i.C_i$ does not subsume the conjunction of all those conjuncts of $C$ other than $\geq k_i R_i.C_i$, and hence $\not\models C\sqsubseteq D$.

For item (ii), let $D\in\lang{\exists,\geq,\sqcap,\top}$ be a concept strictly subsuming $C$, i.e. $\models C\sqsubseteq D$ and $\not\models D\sqsubseteq C$. We show that there is a concept $C'\in F(C)$ such that $C'\sqsubseteq D$. Since $\not\models D\sqsubseteq C$, there is some conjunct of $C$ that does not subsume $D$. If this conjunct is a concept name $A_j\in\Sigma_{\mathsf{C}}$ for some $1\leq j\leq k$, then the concept $C'\in F(C)$ obtained from $C$ by throwing away the conjunct $A_j$ is such that $\models C'\sqsubseteq D$.

Otherwise there is a conjunct of $C$ of the form $\geq k_i R_i.C_i$ for some $1\leq i\leq n$ such that $\not\models D\sqsubseteq\geq k_i R_i.C_i$. By Lem.~\ref{lem:ELQsubsumption}, either all conjuncts of $D$ are of the form $\geq k R. D'$ with $k<k_i$, or else for all such conjuncts with $k\geq k_i$ we have $\not\models D'\sqsubseteq C_i$. In the former case, we have $\models\;\geq (k_i-1)R.C_i\;\sqsubseteq\;\geq k R.D$. 
 
Else, consider some conjunct $\geq k R. D'$ of $D$ with $k\geq k_i$ such that $\models\;\geq k_i R_i.C_i\;\sqsubseteq\;\geq k R.D$ but it subsumes no other conjunct of $C$. Then it must be that $k_i=k, \models C_i\sqsubseteq D'$ and $\not\models D'\sqsubseteq C_i$. But then by the inductive hypothesis there is a concept $C''\in F(C_i)$ such that $\models C''\sqsubseteq D'$, so it follows that $\models\;\geq k_i R.C''\;\sqsubseteq\;\geq k R.D'$. Since the concept $D_i\in F(C)$ contains each $\geq k_i R.C''$ for $C''\in F(C_i)$ as a conjunct, we conclude that $\models D_i\sqsubseteq D$.
\end{proof}

\ELQcharacterisations*

\begin{proof}
Let a concept $C\in\lang{\geq,\sqcap,\top}$ be given. We set $E^+$ to be the set of all positive examples $(\Imc,d)$ of $C$ of size $|\Delta^{\Imc}|\leq |C|^{|C|}$, and observe that $|E^+|$ is in $\mathcal{O}(2^{2^{|C|}})$. For the negative examples, using Thm.~\ref{thm:frontier} we construct a polynomial frontier $\mathcal{F}(C)$ for $C$ w.r.t.~$\lang{\geq,\sqcap,\top}$, which consists of polynomially many concepts strictly subsuming $C$, all of polynomial size. For each $C'\in\mathcal{F}(C)$, we have $\not\models C'\sqsubseteq C$ so and let $(\Imc_{C'},d_{C'})$ be a pointed interpretation witnessing this. By Lemmas~\ref{lem:small-subinterpretation} and \ref{lem:subset-monotone}, there is a small sub-interpretation $\Imc'_{C'}\subseteq\Imc_{C'}$ with $d_{C'}\in C^{\Imc'}$ and $|\Imc'_{C'}|\leq|C'|^{|C'|}$. Let $E^-:=\{(\Imc'_{C'},d_{C'}) \mid C'\in\mathcal{F}(C)\}$.

Since $|C'|$ is polynomial in $|C|$, it follows that $|\Imc'_{C'}|$ is polynomial in $|C|$ for all $C'\in\mathcal{F}(C)$. It follows that the sum of the sizes of the negative examples is polynomial in $|C|$, so that $|E^+\cup E^-|$ is in $\mathcal{O}(2^{2^{|C|}})$. 

We claim that $E=(E^+,E^-)$ is a finite
characterisation of $C$ w.r.t.~$\lang{\geq,\sqcap,\top}$.
Consider some concept $D\in\lang{\geq,\sqcap,\top}$ and suppose that $D$ fits $(E^+,E^-)$. Suppose for contradiction that $\not\models C\sqsubseteq D$. From Lemmas~\ref{lem:small-subinterpretation} and \ref{lem:subset-monotone} it follows that there is a pointed interpretation $(\Imc,d)$ of size 
$|\Delta^{\Imc}|\leq|C|^{|C|}$ with $d\in(C\sqcap\neg D)^{\Imc}$.
By definition of $E^+$, we have $(\Imc,d)\in E^+$ but then $D$ does not fit $E^+$ after all. We conclude that $\models C\sqsubseteq D$.

Next suppose for the sake of a contradiction that  $\not\models D\sqsubseteq C$. In that case $D$ is a concept strictly subsuming $C$, so by the definition of frontiers there is some concept $C'\in F(C)$ such that $\models C'\sqsubseteq D$. Now there is a negative example $(\Imc_{C'},d_{C'})\in E^-$ such that $d_{C'}\in (C'\sqcap\neg C)^{\Imc_{C'}}$, and therefore also $d_{C'}\in(D\sqcap\neg C)^{\Imc_{C'}}$. Thus $D$ does not fit $E^-$ after all. Hence, it must be that $C\equiv D$.
\end{proof}

Next, we prove the negative results without ontologies from Section~\ref{sec:withoutontologies}.

\countingcupbot*

\begin{proof}
Let $A\in\Sigma_{\mathsf{C}}$ and suppose for the sake of a contradiction that it has a finite characterisation $(E^+,E^-)$. Let $k$ be any number larger than the size of the domain of each example in $E^+\cup E^-$. We have that $C_k:=A\sqcup\geq k R.A\in\lang{\geq,\sqcup}$ fits $(E^+,E^-)$. Clearly $C_{k}$ fits all positive examples for $A$ by properties of $\sqcup$, and by choice of $k$, the concept $C_k$ cannot be true on any negative example in $(\Imc,d)\in E^-$, because $|\Delta^{\Imc}|< k$. However, $C_k$ is clearly not logically equivalent to $A$, as witnessed by the pointed interpretation $(\Jmc,d)$ consisting of $d$ together with $k$ other points $d_1,\ldots d_k$, where $d\not\in A^{\Jmc}$ but all other points are, $d$ is $R$-related to all $d_j$'s and no other points are related to each other.

Suppose for contradiction that $\bot$ has a finite characterisation $(E^+,E^-)$. By definition of $\bot$, we have $E^+=\emptyset$. Let $k$ be any number greater than the size of the domain of each example and let $C_k=\geq k R.A$, for some $A\in\NC$. Then, $C_k\in\lang{\geq,\bot}$ fits $E^-$ since, for all $(\Imc,d)\in E^-$, we have $|\Delta^{\Imc}|< k$. Hence, as $E^+=\emptyset$, we have that $C_k$ fits $(E^+,E^-)$. However, $C_k$ is not equivalent to $\bot$, since $C_k$ is clearly satisfiable.

Suppose for contradiction that this fragment does admit finite characterisations, and let $A\in\Sigma_{\mathsf{C}}$ be a concept name. Then in particular, the concept $A$ must have a finite characterisation $(E^+, E^-)$. By Lemma \ref{lem:canposexamples} it must be that any $\lang{\forall,\exists,-,\sqcap}$ concept not entailed by $A$ must be falsified on some positive example in $E^+$. For each $n$, let $C_n:=[\forall R.]^n[\exists R^-.]^n[\exists R.]^{n+1}[\exists R^-.]^{n+1}A$. We show, for all $n$, that $\not\models A\sqsubseteq C_n$ and $\models A\sqcap\neg C_n\sqsubseteq height^{R}_n$. First of all, let $(\Imc_n,d)$ be the interpretation consisting of a directed $R$-path of length $n$, where $d\in A^{\Imc_n}$; i.e.
\[d\xrightarrow{R} d_2\xrightarrow{R} \cdots \xrightarrow{R} d_{n+1}\]
We have $d\in(A\sqcap\neg C_n)^{\Imc_n}$ where $\neg C_n\equiv[\exists R.]^n[\forall R^-.]^n[\forall R.]^{n+1}[\forall R^-.]^{n+1}.\neg A$. First, of all, we have $d\in A^{\Imc_n}$ by assumption. Further, clearly $d\in([\forall R]^{n+1}\bot)^{\Imc_n}$ but as $[\forall R.]^{n+1}\bot$ subsumes $[\forall R.]^{n+1}[\forall R^-.]^{n+1}\neg A$, it follows that $d\in([\forall R]^{n+1}[\forall R^-.]^{n+1}\neg A)^{\Imc_n}$.

Next, we show that $\models A\sqcap\neg C_n\sqsubseteq height^{R}_n$ for all $n$. So pick any $n$ and let $d\in(A\sqcap\neg C_n)^{\Imc}$. Recall $\neg C_n\equiv[\exists R.]^n[\forall R^-.]^n[\forall R.]^{n+1}[\forall R^-.]^{n+1}\neg A$, so there is a directed $R$-path of length $n$ from $d$ to a state $d'\in\Delta^{\Imc}$ such that for all directed $R^-$-paths of length $n$ emanating from $d'$ end in states satisfying $[\forall R]^{n+1}[\forall R^-]^{n+1}\neg A$. Hence in particular we have $d\in(A\sqcap[\forall R]^{n+1}[\forall R^-]^{n+1}\neg A)^{\Imc}$ but this concept is subsumed by $[\forall R.]^n\bot$, so it must be that $d\in([\forall R.]^{n+1}\bot)^{\Imc}$. Moreover, clearly $\models\neg C_n\sqsubseteq[\exists R.]^n\top$ and hence $d\in(height^{R}_n)^{\Imc}$.

Now choose $k$ to be any number strictly larger than the size of any example in $E^+$. Then $C_k$ fits $E^+$, since no positive example in $E^+$ can satisfy $A\sqcap\neg C_k$ as that would imply that it also satisfies $height^{R}_k$, which is impossible by choice of $k$. It follows that $A\sqcap C_k$ fits the entire set $(E^+,E^-)$, yet $A\not\equiv A\sqcap C_k$ since $\not\models A\sqsubseteq C_k$.
\end{proof}

This concludes the proofs of all the results needed to show Main Thm.~\ref{thm:mainone} which provides an almost complete classification of which fragments of $\mathcal{ALCQI}$ admit finite characterisations.

We continue with the results needed to show Main Thm.~\ref{thm:maintwo}, which provides a complete classification of which fragments containing at least $\exists$ and $\sqcap$ of $\mathcal{ALCQI}$ admit polynomial time computable characterisations. We prove two new negative results, using two Lemmas.

\lemcanposexample*

\begin{proof}
Let $C,C'\in\lang{\Obf}$ with $\sqcap\in\Obf$ such that $\not\models C\sqsubseteq C'$. Suppose for contradiction that $(E^+,E^-)$ is a finite characterisation of $C$ w.r.t. $\lang{\Obf}$ and that $C'$ is not falsified on any positive example in $E^+$. Then it follows that $C\sqcap C'$ fits all the positive examples, and also fits all negative examples in $E^-$ because the conjunct $C$ is false on all of them. However, since $C\not\equiv C'$, this is a contradiction.
\end{proof}

\ELQlowerbound*

\begin{proof}
Let $n$ be a natural number, $\Sigma_{\mathsf{C}}:=\{A_{i1},A_{i2},A_{i3}\mid 1\leq i\leq n\}$ and $\Sigma_{\mathsf{R}}=\{R\}$. For each $1\leq i\leq n$ define the concepts $C_i:=\exists R.(A_{i1}\sqcap A_{i3})\sqcap\exists R.(A_{i2}\sqcap A_{i3})$, $D_i:=\exists R.(A_{i1}\sqcap A_{i2}\sqcap A_{i3})$ and $D'_i:=\geq 2 R.A_{i3}$. Observe that $\not\models C_i\sqsubseteq D_i$ and $\not\models C_i\sqsubseteq D'_i$ yet $\models C_i\sqsubseteq D_i\sqcup D'_i$. In other words $\models C_i\sqcap\neg D_i\sqsubseteq D'_i$ and $\models C_i\sqcap\neg D'_i\sqsubseteq D_i$. This is because every point satisfying $C_i$ in some interpretation must have at least one successor satisfying $A_{i1}\sqcap A_{i2}\sqcap A_{i3}$, or at least two distinct successors satisfying $A_{i3}$. 

Moreover, none of the concepts in $\{D_1,D_1',\ldots,D_n,D'_n\}$ subsume one another since $D_i$ and $D'_i$ do not subsume one another for all $1\leq i\leq n$, and for every $1\leq j\leq n$ with $j\ne i$ we have that $D_i$ and $D'_i$ do not share any common concept names with $D_j$ and $D'_j$, and hence do not subsume one another. Define
\[C:=\underset{1\leq i\leq n}{\bigsqcap}\exists R.(C_i\sqcap\underset{1\leq j\leq n,\;j\ne i}{\bigsqcap}(D_j\sqcap D'_j))\]
and suppose that $C$ has a finite characterisation $(E^+,E^-)$ w.r.t. $\lang{\leq,\sqcap}$. We use Lemma \ref{lem:canposexamples} to show that $|E^+|\geq 2^{\mathcal{O}(|C|)}$. This establishes an exponential lower bound on size of finite characterisations for the concept $C$ w.r.t. $\lang{\geq,\sqcap}[\Sigma_{\mathsf{C}},\Sigma_{\mathsf{R}}]$ in terms of $|C|$ (and the size of the 
signature). This lower bound still holds when number are represented in unary, because the number restriction in the concept $C$ (which is implicitly parameterised by $n$) is kept constant.

For every $\sigma\in\{D_1,D'_1\}\times\ldots\times\{D_n,D'_n\}$ define
\[D_{\sigma}:=\exists R.(\underset{1\leq i\leq n}{\bigsqcap}\sigma(i))\]
where $\sigma(i)\in\{D_i,D'_i\}$ is the $i$'th component of the sequence $\sigma$. That is, $D_{\sigma}$ is true if there is at least one $R$-successor satisfying every concept occurring in $\sigma$. Clearly, there are $2^{n}$ many distinct concepts of the form $D_{\sigma}$ because each $\sigma$ determines a unique subset of $\{D_1,D'_1,\ldots,D_n,D'_n\}$. We saw before that none of the concepts in this set entail each other, and hence $D_{\sigma}$ does not entail $D_{\sigma'}$ for any $\sigma\ne\sigma'$.

We show for each $\sigma$ that (i) $\not\models C\sqsubseteq D_{\sigma}$ and (ii) $\models C\sqcap\neg D_{\sigma}\sqsubseteq D_{\sigma'}$ for all $\sigma'\ne\sigma$. For then by (i) every $D_{\sigma}$ must be falsified on positive example in $E^+$, and by (ii) this example has to be distinct from any other positive example chosen to falsify $D_{\sigma'}$ for any $\sigma'\ne\sigma$. Thus, it follows from (i) and (ii) that $|E^+_|\geq 2^n$.

By Lem.~\ref{lem:ELQsubsumption}, the subsumption in (i) holds iff there is some $1\leq i\leq n$ such that
\[\models C_i\sqcap\underset{1\leq j\leq n,\;j\ne i}{\bigsqcap}(D_j\sqcap D'_j)\sqsubseteq\sigma(i)\]
since for all $j\ne i$ we already know that $\sigma(j)\in\{D_j,D'j\}$ occurs as a conjunct on the left hand side. As noted before, it cannot be that $\sigma(i)\in\{D_i,D'_i\}$ subsumes either $D_j$ or $D_j'$ for $j\ne i$ since they share no common concept names. Thus the above subsumption holds iff $\models C_i\sqsubseteq\sigma(i)$, which we know is not the case as $\not\models C_i\sqsubseteq D_i$ and $\not\models C_i\sqsubseteq D'_i$. Therefore the subsumption cannot hold which proves item (i).

For item (ii), let $\sigma,\sigma'\in\{D_1,D'_1\}\times\ldots\times\{D_n,D'_n\}$ such that $\sigma\ne\sigma'$, i.e. there is some $1\leq k\leq n$ such that without loss of generality $\sigma(k)=D_k\ne D'_k=\sigma'(k)$. Note that 
\[\neg D_{\sigma}\equiv\forall R.(\underset{1\leq i\leq n}{\bigsqcup}\neg\sigma(i))\]
It follows that $\models C\sqcap\neg D_{\sigma}\sqsubseteq\exists R.(C_i\sqcap\neg D_k)$. For suppose that $d\in(C\sqcap\neg D_{\sigma})^{\Imc}$ for some interpretation $\Imc$. It follows that there is some $e\in R^{\Imc}[d]$ such that $d\in C_k^{\Imc}$ and $d\in(D_j\sqcap D'_j)^{\Imc}$ for all $j\ne i$ but where $d$ falsifies at least one component of $\sigma$. Clearly this can only be the case if $d\in(\neg D_k)^{\Imc}$, since we assumed that $\sigma(k)=D_k$, and every other component of $\sigma$ is in $\{D_1,D'_1,\ldots, D_{k-1},D'_{k-1},D_{k+1},D'_{k+1},\ldots,D_n,D'_n\}=\{D_j,D'_j\;|\;1\leq j\leq n, j\ne k\}$. But since $\models C_k\sqcap\neg D_k\sqsubseteq D'_k$, it follows that $d\in (D'_k)^{\Imc}=(\sigma'(k))^{\Imc}$ and hence $d$ satisfies every component of $\sigma'$. This shows that $\models C\sqcap\neg D_{\sigma}\sqsubseteq D_{\sigma'}$.
\end{proof}

While we managed to establish an explicit exponential lower bound for the size of finite characterisations for $\lang{\geq,\sqcap}$,in the following we show that it is NP-hard to compute finite characterisations for $\lang{\forall,\exists,\sqcap}$. The following Lemma then helps us to prove that this fragment does not admit polynomial time computable characterisations conditional on common complexity-theoretic assumptions.

\lempolycharimpliessubsumptioninp*
\begin{proof}
Given $\lang{\Obf}$ concepts $C$ and $D$, we can decide whether $C$ entails $D$ by computing in polynomial time a finite characterisation $(E^+,E^-)$ for $C$ w.r.t. $\lang{\Obf}$. By Lemma \ref{lem:canposexamples} it follows that $C$ entails $D$ iff $D$ fits $E^+$. We can check in polynomial time whether $D$ fits $E^+$ since model checking $\lang{\Obf}$ can be done in polynomial time.
\end{proof}

\NPhardness*

\begin{proof}
It was shown in \cite{DONINI1992} that the description logic $\lang{\forall,\exists,\sqcap,\top}$ has an NP-complete subsumption problem, already over the signature $\Sigma_{\mathsf{C}}=\emptyset$ and $\Sigma_{\mathsf{R}}=\{R\}$. We give a polynomial time reduction from subsumption in $\lang{\forall,\exists,\sqcap,\top}$ to $\lang{\forall,\exists,\sqcap}$ over an extended signature with a single concept name $B$. We use $B$ to simulate $\top$ in relevant submodels. 

Consider any two concepts $C,D\in\lang{\forall,\exists,\sqcap,\top}[\emptyset,\{R\}]$. Define $C_0$ to be the conjunction of all concepts of the form $[\forall R.]^n B$ where $B$ is a fresh concept name, $n\leq k$ and $k:=\text{max}\{\mathsf{dp}(C),\mathsf{dp}(D)\}$. Clearly $|C_0|$ is of size polynomial in $|C|$ and $|D|$ and the size of the signature. In fact it is not hard to see that $|C_0|$ is in $\mathcal{O}(k^2)$.

Further, let $C', D'$ be the result of substituting $B$ for every occurrence of $\top$ in $C$ and $D$, respectively. Finally, let $C''$ be the conjunction of $C'$ with $C_0$, and similarly let $D''$ be the conjunction of $D'$ with $C_0$. Note that $\models C''\sqsubseteq C\sqcap C_0$. We claim that $\models C\sqsubseteq D$ iff $\models C''\sqsubseteq D''$. This proves our claim as $C'',D''\in\lang{\forall,\exists,\sqcap}[\{B\},\{R\}]$.

From left to right, suppose that $\models C\sqsubseteq D$ and let $d\in (C'')^{\Imc}$. Since $\models C''\sqsubseteq C\sqcap C_0$ it follows that $C''\sqsubseteq D\sqcap C_0$ as well. But by choice of $k\geq\mathsf{dp}(D)$ we get that $\models D\sqcap C_0\sqsubseteq D'$ because $B$ is forced to be true in all states reachable from $d$ (following $R$-edges) in at most $k$ steps. Since $D''=D'\sqcap C_0$, chasing the subsumptions established above we conclude that $\models C''\sqsubseteq D''$. From right to left, suppose that $\models C''\sqsubseteq D''$, and let $d\in C^{\Imc}$ for some interpretation $\Imc$. Expand $\Imc$ to an interpretation $\Imc'$ over the extended signature by making $B$ true everywhere (all other things being equal). It follows that $d\in(C'\sqcap C_0)^{\Imc}=(C'')^{\Imc'}$. Applying the subsumption $\models C''\sqsubseteq D''$ we get $d\in(D'')^{\Imc}$. This shows that $\models C\sqsubseteq D$.
\end{proof}

\section{Proofs for Section~\ref{sec:withontologies}}\label{app:proofsofsecfour}

To simplify our proofs, we assume that all ontologies are in \emph{named form},
which we can assume without loss of generality by Lemma~\ref{lem:namedform}. 




 \begin{definition}[(Simulation)]\label{def:simulation}
    Let $(\Imc_1,d_1)$ and $(\Imc_2,d_2)$ be pointed interpretations. A \emph{simulation}
    $Z:(\Imc_1,d_1) \mathop{\underline{\to}}  
      (\Imc_2,d_2)$ is a binary relation $Z\subseteq \Delta^{\Imc_1}\times \Delta^{\Imc_2}$ which contains 
      $(d_1,d_2)$ and such that for all 
      $(d,d')\in Z$, 
      (i) for all $A\in \NC$,
      if $d\in A^{\Imc_1}$ then $d'\in A^{\Imc_2}$;
      (ii) for all $R\in \NR$ and pairs $(d,e)\in R^{\Imc_1}$ there is a pair $(d',e')\in R^{\Imc_2}$ with $(e,e')\in Z$.
\end{definition}


Before we provide the definition of the canonical model, 
we introduce the following notions.
Given an \EL concept $D$,
we inductively define the tree-shaped interpretation $\Imc_D$ of $D$, with the root denoted $d_D$, as follows. 
When $C$ is $\top$, we define $\Imc_{\top}$ as the interpretation with $\Delta^{\Imc_{\top}}:=\{d_{\top}\}$  and all
  concept and role names interpreted as the empty set.
For $D$ a concept name $A\in \NC$ we define $\Imc_A$ as the interpretation with $\Delta^{\Imc_A}:=\{d_A\}$,
$A^{\Imc_A}:=\{d_A\}$, and all
other concept and role names interpreted as the empty set.
For $D=\exists R.C$, we define
$\Imc_D$ as the interpretation with $\Delta^{\Imc_D}:=\{d_D\}\cup \Delta^{\Imc_C}$,
all concept and role name interpretations are as for
$\Imc_C$ except that we replace
$d_C$ by $d_{C\sqcap \exists R^-}$,
add $(d_D, d_{C\sqcap \exists R^-})$ to $R^{\Imc_D}$, 
and assume $d_D$ is fresh (i.e., it is not in $\Delta^{\Imc_C}$).  Finally, for $D=D_1\sqcap D_2$
we define $\Delta^{\Imc_D}:=\Delta^{\Imc_{D_1}}\cup (\Delta^{\Imc_{D_2}}\setminus\{d_{D_2}\})$,
assuming $\Delta^{\Imc_{D_1}}$ and $\Delta^{\Imc_{D_2}}$
are disjoint, and
with
all concept and role name interpretations   as in
$\Imc_{D_1}$ and $\Imc_{D_2}$, except that
we connect $d_{D_1}$ with the elements of $\Delta^{\Imc_{D_2}}$
in the same way as $d_{D_2}$ is connected. In other words, we 
 {identify} $d_{D_1}$ with the root $d_{D_2}$ of $\Imc_{D_2}$. 
The following is clear from the definition of $\Imc_D$, rooted in $d_D$.

\begin{proposition}\label{prop:tree}
Given an \EL concept $C$, for all $d_D\in\Delta^{\Imc_C}$, $d_D\in E^{\Imc_C}$ iff $\models D\sqsubseteq E$, where $E$ is an \EL concept or of the form $\exists R^-$, with $R\in\NR$. 
\end{proposition}

Given a DL-Lite  ontology \Omc and an \EL concept $D$, 
let ${\sf sig}(\Omc)$ be the set of concept and role names occurring in \Omc. Denote by $\Delta^{{\sf sig}(\Omc)}$ the set \[\{d_{\exists R},d_{\exists R^-} \mid R \in \NR\cap{\sf sig}(\Omc), \exists R \text{ is satisfiable w.r.t. } \Omc \}\  \] 
and assume it to be disjoint from $\Delta^{\Imc_D}$. 

\begin{definition}\label{def:can-model}
  Let \Omc be a  DL-Lite  ontology in named form and $D$ an \EL concept     satisfiable w.r.t. \Omc.   The canonical model $\Imc_{D,\Omc}$ for   $\Omc$ and  $D$
   is: 
   \begin{itemize}
       \item $\Delta^{\Imc_{D,\Omc}} := 
       \Delta^{\Imc_D} \cup \Delta^{{\sf sig}(\Omc)}$,
\item $A^{\Imc_{D,\Omc}} := 
\{d_C \in \Delta^{\Imc_{D,\Omc}} {\mid} \Omc\models  C \sqsubseteq A \}$
for all $A\in \NC$;
\item $R^{\Imc_{D,\Omc}}:= \{(d_C,d_{C'})\in (\Delta^{\Imc_D})^2 \mid (d_C,d_{C'})\in R^{\Imc_D}\} 
\cup \\ \{(d_C,d_{\exists R^-})\in \Delta^{\Imc_{D,\Omc}} \times \Delta^{{\sf sig}(\Omc)} \mid  \Omc\models C\sqsubseteq \exists R\}\cup \\ \{(d_{\exists R},d_{C})\in \Delta^{{\sf sig}(\Omc)} \times \Delta^{\Imc_{D,\Omc}} \mid \Omc\models C\sqsubseteq  \exists R^-\}$, for all $R\in \NR$.
    \end{itemize}

\end{definition}

\begin{restatable} 
{lemma}{lemmacan}
\label{lem:lemmacan}
\label{lem:canonical}
Let $C$ be any $\lang{\sqcap,\exists,\top}$-concept expression and 
  let $\Omc$ be any  
  DL-Lite  ontology in named form such that $C$ is satisfiable w.r.t. $\Omc$.
  Let $\Imc_{C,\Omc}$ be as in Definition~\ref{def:can-model}.
  Then, for all pointed interpretations $(\Imc,d)$ satisfying $\Omc$, 
  the following are equivalent:
  \begin{enumerate}
      \item $d\in C^\Imc$ and
      \item there is a simulation $Z:(\Imc_{C,\Omc},d_{C}) \mathop{\underline{\to}} (\Imc,d)$.
  \end{enumerate}
\end{restatable}
\begin{proof}
We first argue that the lemma holds when $\Omc=\emptyset$. 


\begin{claim}\label{claim:elsimulation}
    Let \Imc be an interpretation. For all $d\in \Delta^\Imc$ and all \EL concepts $C$,
     $d\in C^\Imc$ iff there is a simulation $Z:(\Imc_C,d_C)\mathop{\underline{\to}}  (\Imc,d)$.
\end{claim}
\begin{proof}[Proof of Claim~\ref{claim:elsimulation}]
  The proof is by induction. In the base cases, $C=A$ where $A\in\NC$ or $C=\top$. In these cases,  by definition
of $\Imc_C$ and the definition of simulation, $d\in C^\Imc$ iff $\{(d_C,d)\}$ is a simulation from $(\Imc_C,d_C)$ to $(\Imc,d)$, as required.
We now proceed with the inductive step. Suppose the claim holds for $C_1,C_2$. We make a case distinction.
\begin{itemize}
    \item $C=C_1\sqcap C_2$. 
    If
    $d\in (C_1\sqcap C_2)^\Imc$ then, by the semantics of \EL, $d\in C^\Imc_1$ and $d\in C^\Imc_2$.
    Let $\Imc_{C_1}$ and $\Imc_{C_2}$ be   tree-shaped interpretations of $C_1$ and $C_2$, respectively, assuming $\Delta^{\Imc_{C_1}}$ and $\Delta^{\Imc_{C_2}}$ are disjoint.
    By the inductive hypothesis, if $d\in C^\Imc_1$ and $d\in C^\Imc_2$ then
       there is a simulation $Z_1:(\Imc_{C_1},d_{C_1})\mathop{\underline{\to}}  (\Imc,d)$ and a simulation
     $Z_2:(\Imc_{C_2},d_{C_2})\mathop{\underline{\to}}  (\Imc,d)$. 
     Let
     $Z'_1$ and $Z'_2$ be the result of replacing the occurrences of  $d_{C_1}$ and $d_{C_2}$, respectively, in $Z_1$ and $Z_2$
     by $d_{C_1\sqcap C_2}$. We have that
     $Z:=Z'_1\cup Z'_2$ is a simulation
     from $(\Imc_{C_1\sqcap C_2},d_{C_1\sqcap C_2})$ to $(\Imc,d)$. Conversely,
     suppose $Z$ is a simulation from $(\Imc_{C_1\sqcap C_2},d_{C_1\sqcap C_2})$ to $(\Imc,d)$.
     Then, by definition of $\Imc_{C_1\sqcap C_2}$, there are simulations $Z_1:(\Imc_{C_1},d_{C_1})\mathop{\underline{\to}}  (\Imc,d)$ and $Z_2:(\Imc_{C_2},d_{C_2})\mathop{\underline{\to}}  (\Imc,d)$. Then, by the inductive hypothesis, $d\in C^\Imc_1$ and $d\in C^\Imc_2$. So $d\in (C_1\sqcap C_2)^\Imc$  by the semantics of \EL.
    \item $C=\exists R.C_1$. Suppose $d\in (\exists R.C_1)^\Imc$. By the semantics of \EL, there is $d'\in\Delta^\Imc$ such that $(d,d')\in R^\Imc$ and $d'\in C^\Imc_1$. By the inductive hypothesis,
       there is a simulation $Z_1:(\Imc_{C_1},d_{C_1})\mathop{\underline{\to}}  (\Imc,d')$. We have that $Z:=Z_1\cup \{(d_{\exists R.C_1},d)\}$ is a simulation
     from $(\Imc_{\exists R.C_1},d_{\exists R.C_1})$ to $(\Imc,d)$, as required. Conversely, suppose $Z$ is a simulation
     from $(\Imc_{\exists R.C_1},d_{\exists R.C_1})$ to $(\Imc,d)$. By Proposition~\ref{prop:tree},
     $d_{\exists R.C_1}\in ({\exists R.C_1})^{\Imc_{\exists R.C_1}}$ (since $\models \exists R.C_1\sqsubseteq \exists R.C_1$ trivially). By the semantics of \EL and the definition of ${\Imc_{\exists R.C_1}}$, there is $d_{C_1}\in\Delta^{\Imc_{\exists R.C_1}}$ such that $(d_{\exists R.C_1},d_{C_1})\in R^{\Imc_{\exists R.C_1}}$ and $d_{C_1}\in C^{\Imc_{\exists R.C_1}}_1$. Since $Z$ is a simulation and $(d_{\exists R.C_1},d)\in Z$, there is $d'\in\Delta^\Imc$ such that
     $(d,d')\in R^\Imc$ and $(d_{C_1},d')\in Z$. We have that  $Z_1:=Z\setminus \{(d_{\exists R.C_1},d)\}$ is a simulation from 
     $(\Imc_{C_1},d_{C_1})$ to $(\Imc,d')$.
 By the inductive hypothesis, $d'\in C^\Imc_1$. Then, by the semantics of \EL, $d\in (\exists R.C_1)^\Imc$. 
\end{itemize}
\end{proof}
We now consider the case in which \Omc is an arbitrary DL-Lite  ontology.
The direction ($\Leftarrow$), follows from the fact that $\Imc_C$ is sub-interpretation of $\Imc_{C,\Omc}$.
So if there is a simulation from $(\Imc_{C,\Omc},d_C)$ to $(\Imc,d)$ then 
there is also simulation from $(\Imc_{C},d_C)$ to $(\Imc,d)$.
Then, by Claim~\ref{claim:elsimulation},
$d\in C^\Imc$.
To prove $(\Rightarrow)$, we
  use Definition~\ref{def:can-model}, Claim~\ref{claim:elsimulation}, and the following technical claims to establish the lemma. 
\begin{claim}\label{claim:aux}
    Let \Omc be a DL-Lite  ontology
    in named form and let $C$ be an \EL concept or a concept of the form $\exists R^-$. Then, $\Omc\models C\sqsubseteq \exists S$, where $S=R^-$ or $S=R$ with $R\in\NR$, iff $\models C\sqsubseteq \exists S$ or there is
    $A\in\NC$ such that $\Omc\models C\sqsubseteq A$ and $\Omc\models A\sqsubseteq \exists S$.
\end{claim}
\begin{proof}
The direction ($\Leftarrow$) is immediate. Regarding ($\Rightarrow$), suppose $\Omc\models C\sqsubseteq \exists S$. If $\models C\sqsubseteq \exists S$ we are done. Then assume this is not the case. Since $\Omc$ is in named form, all 
concept inclusions in \Omc with $\exists S$ on the right side  need to have a concept name on the left side.  As $\not\models C\sqsubseteq \exists S$, the only way to derive $\exists S$ is by one of these concept inclusions.
So  there is
    $A\in\NC$ such that $\Omc\models C\sqsubseteq A$ and $\Omc\models A\sqsubseteq \exists S$.
\end{proof}
Let $C$ be an \EL concept and let $\Imc'_{C,\Omc}$ be 
the sub-interpretation of
$\Imc_{C,\Omc}$ (Definition~\ref{def:can-model}) that results from  removing
all elements not reachable 
from the elements of $\Imc_C$ via a chain of role names.
\begin{claim}\label{claim:sim}
   If $Z$ is a simulation from $(\Imc'_{C,\Omc},d_C)$ to 
   $(\Imc,d)$ then $Z$ is a simulation from $(\Imc_{C,\Omc},d_C)$ to 
    $(\Imc,d)$.
\end{claim}
\begin{proof}
  By  the definition of a simulation, Condition (i) of Definition~\ref{def:simulation} is satisfied.  Condition (ii) of Definition~\ref{def:simulation} is satisfied due to
  the definition of $\Delta^{\Imc'_{C,\Omc}}$,
  which contains all elements of $\Delta^{\Imc_{C,\Omc}}$ that reachable with a chain of role names (but excludes those potentially disconnected and those that only reachable using inverse roles).
\end{proof}
 So from now on we concentrate on the sub-interpretation $(\Imc'_{C,\Omc},d_C)$. This step is important because in the canonical model we make non-empty all DL-Lite concepts  over ${\sf sig}(\Omc)$ that are satisfiable but they may be empty in a model. So there may be no simulation from $(\Imc_{C,\Omc},d_C)$ to $(\Imc,d)$ that covers all elements in   $\Delta^{\Imc_{C,\Omc}}$. Moreover, only elements of the form $d_{\exists R^-}$ with $R\in\NR$ are   reachable from the elements in $\Delta^{\Imc_C}$ via a chain of role names. We formalize this in Claim~\ref{claim:crucial}.
\begin{claim}\label{claim:crucial} For all $R\in\NR$ and all $A\in\NC$, elements of the form 
    $d_{\exists R} \in \Delta^{{\sf sig}(\Omc)}$ are not in $ \Delta^{\Imc'_{C,\Omc}}$.
\end{claim}
\begin{proof}
    By Definition~\ref{def:can-model},
    elements in $\Delta^{\Imc_C}$ 
    are only directly connected via a role name to elements of the form $d_{\exists R^-}$, with $R\in\NR$. Again by Definition~\ref{def:can-model}, elements of the form $d_{\exists S^-}$ are only connected via a role name to elements of the form $d_{\exists R^-}$ (in  Definition~\ref{def:can-model}, this is  the second set in the second item, where we have the pair $(d_C,d_{\exists R^-})$ but no pair of the form $(d_C,d_{\exists R})$.
\end{proof}
Let \Imc be an interpretation and let $d$ be an arbitrary element in $\Delta^\Imc$.
If   $d \in C^\Imc$ then by Claim~\ref{claim:elsimulation} there is a simulation $Z:(\Imc_C,d_C)\mathop{\underline{\to}}  (\Imc,d)$.  We argue that if $\Imc\models \Omc$ then one can extend $Z$ and define
  a simulation from $(\Imc'_{C,\Omc},d_C)$ to $(\Imc,d)$.
  Let $Z'$ be the union of $Z$ and 
  \[\{(d_{\exists R^-},e)\in \Delta^{\Imc'_{C,\Omc}}\times \Delta^\Imc \mid e\in (\exists R^-)^\Imc\},\]
  where  $R\in\NR$. Note that by the definition of $\Imc'_{C,\Omc}$ and Claim~\ref{claim:crucial}, $Z'$ covers the whole domain of $\Imc'_{C,\Omc}$.
In the next two claims we show that $Z'$ is a simulation 
from $(\Imc'_{C,\Omc},d_C)$ to $(\Imc,d)$.
  First we show Condition (i) of Definition~\ref{def:simulation}. 
\begin{claim}\label{claim:simulationconcept}
For all $A\in \NC$,
if $d_D\in A^{\Imc'_{C,\Omc}}$ and $(d_D,e)\in Z'$ then
$e\in A^\Imc$.    
\end{claim}
\begin{proof}
We make a case distinction based
on the elements of $\Delta^{\Imc'_{C,\Omc}}$ which has elements belonging to $\Delta^{\Imc_C}$ and elements belonging to $\Delta^{{\sf sig}(\Omc)}$.
\begin{itemize}
    \item $d_D\in \Delta^{\Imc_C}$. If $d_D \in A^{\Imc'_{C,\Omc}}$ then, by Definition~\ref{def:can-model},
    either $d_D \in A^{\Imc_{C}}$ or
    $\Omc\models D\sqsubseteq A$. In the former case, since $Z$ is a simulation,
    $e\in A^\Imc$. In the latter, by definition of $\Imc_C$ and Proposition~\ref{prop:tree}, $d_D\in D^{\Imc_C}$ (since $\models D\sqsubseteq D$ trivially). Now, since $(d_D,e)\in Z\subseteq Z'$ and $Z$ is a simulation, by Claim~\ref{claim:elsimulation}, $e\in D^\Imc$. As $\Imc\models\Omc$ and $\Omc\models D\sqsubseteq A$, we have that $e\in A^\Imc$.
    \item $d_D \in \Delta^{{{\sf sig}(\Omc)}}$.
    If $d_D\in A^{\Imc'_{C,\Omc}}$ then, by Definition~\ref{def:can-model},  $\Omc\models D\sqsubseteq A$. Also, by definition of $Z'$, $e\in D^\Imc$.
    Since $\Imc\models\Omc$, we have that
    $e\in A^\Imc$.
\end{itemize}
\end{proof}
We now show    Condition (ii) of  Definition~\ref{def:simulation}.  
\begin{claim}\label{claim:simulationrole}
For all $R\in \NR$,
if $(d_D,d_{D'})\in R^{\Imc'_{C,\Omc}}$ 
and $(d_D,e)\in Z'$ then there is
$e'\in \Delta^\Imc$ such that
$(e,e')\in R^\Imc$ and $(d_{D'},e')\in Z'$.    
\end{claim}
\begin{proof}
We   make a case distinction based
on the pairs of elements in $(\Delta^{\Imc'_{C,\Omc}})^2$ defined to be in 
$R^{\Imc'_{C,\Omc}}$. 
\begin{itemize}
    \item $(d_D,d_{D'})\in \Delta^{\Imc_C}\times \Delta^{\Imc_C}$. If $(d_D,d_{D'}) \in R^{\Imc'_{C,\Omc}}$ then, by Definition~\ref{def:can-model},
      $(d_D,d_{D'}) \in R^{\Imc_{C}}$. 
      By assumption $(d_D,e)\in Z'$ and
      as $(d_D,d_{D'})\in \Delta^{\Imc_C}\times \Delta^{\Imc_C}$, in fact $(d_D,e)\in Z$.
      Since $Z$ is a simulation, there is
$e'\in \Delta^\Imc$ such that
$(e,e')\in R^\Imc$ and $(d_{D'},e')\in Z\subseteq Z'$.
\item $(d_D,d_{\exists R^-})\in \Delta^{\Imc'_{C,\Omc}}\times \Delta^{{\sf sig}(\Omc)}$. If $(d_D,d_{\exists R^-}) \in R^{\Imc'_{C,\Omc}}$
then, by Definition~\ref{def:can-model}, $\Omc\models D\sqsubseteq \exists R$.
By Claim~\ref{claim:aux}, either (a)
    $\models D\sqsubseteq \exists R$ or (b) there is 
$A\in\NC$ such that
$\Omc\models D\sqsubseteq A$ and $\Omc\models A\sqsubseteq \exists R$. In case (a),
we have that $d_D\in \Delta^{\Imc_C}$
(otherwise $d_D\in \Delta^{{\sf sig}(\Omc)}$ and, by Claim~\ref{claim:crucial},
$d_D\in \Delta^{\Imc'_{C,\Omc}}$ must be of the form $d_{\exists S^-}$ but $\not\models \exists S^-\sqsubseteq \exists R$ for all $R\in\NR$). If $d_D\in \Delta^{\Imc_C}$ then in fact $(d_D,e)\in Z\subseteq Z'$. Also, by definition of $\Imc_C$ and Proposition~\ref{prop:tree}, $d_D\in D^{\Imc_C}$ (since $\models D\sqsubseteq D$ trivially). Since $Z$ is a simulation, by Claim~\ref{claim:elsimulation}, $e\in D^\Imc$. As $\models D\sqsubseteq \exists R$, $e\in (\exists R)^\Imc$. Then there 
is $e'\in \Delta^\Imc$ such that $(e,e')\in R^\Imc$. Moreover, $e'\in (\exists R^-)^\Imc$. Then, by definition of $Z'$, $(d_{\exists R^-},e')\in Z'$.
In case~(b),
$d_D\in A^{\Imc'_{C,\Omc}}$.
By Claim~\ref{claim:simulationconcept},
$e\in A^\Imc$. As $\Imc\models \Omc$ and $\Omc\models A\sqsubseteq \exists R$,
there is $e'\in \Delta^\Imc$ such that
$(e,e')\in R^\Imc$. This means that
$e'\in (\exists R^-)^\Imc$.
By definition of $Z'$,
$(d_{\exists R^-},e')\in Z'$.
    \item $(d_{\exists R},d_{D'})\in \Delta^{{\sf sig}(\Omc)} \times \Delta^{\Imc'_{C,\Omc}}$. 
    By Claim~\ref{claim:crucial}, 
    there is no $d_{\exists R}\in \Delta^{{\sf sig}(\Omc)}$ with $R\in \NR$ and $d_{\exists R}\in \Delta^{\Imc'_{C,\Omc}}$. In other words
    that the case $(d_{\exists R},d_{D'}) \in R^{\Imc'_{C,\Omc}}$ cannot happen (even though $(d_{\exists R},d_{D'}) \in R^{\Imc_{C,\Omc}}$ is possible).
 %
\end{itemize}
\end{proof}
By Claims~\ref{claim:simulationconcept} and~\ref{claim:simulationrole}, $Z'$ is a simulation from $(\Imc'_{C,\Omc},d_C)$ to $(\Imc,d)$. 
Finally, by Claim~\ref{claim:sim}, $Z'$ is in fact also a simulation from $(\Imc_{C,\Omc},d_C)$ to $(\Imc,d)$.
\end{proof}

We have just shown Lemma~\ref{lem:lemmacan}, so now we are ready to prove Theorem~\ref{thm:canmod}.


\thmcanmod*
\begin{proof}
Let $\Imc_{D,\Omc}$ be as in Definition~\ref{def:can-model}.
Assume w.l.o.g. that \Omc is in named form~(Lemma~\ref{lem:namedform}).
We start with Point~1 assuming Point~2 holds. 
Suppose $d_{D}\in C^{\Imc_{D,\Omc}}$. Since,
by Point 2, $\Imc_{D,\Omc}\models\Omc$, we can apply Lemma~\ref{lem:canonical}.
By Lemma~\ref{lem:canonical}, there is a simulation $Z:(\Imc_{C,\Omc},d_{C}) \mathop{\underline{\to}} (\Imc_{D,\Omc},d_D)$. Let \Imc be an arbitrary interpretation such that $\Imc\models\Omc$ and let $d$ be an arbitrary element of $\Delta^\Imc$ such that $d\in D^\Imc$. By Lemma~\ref{lem:canonical}, there is a simulation $Z':(\Imc_{D,\Omc},d_{D}) \mathop{\underline{\to}} (\Imc,d)$.
Since simulations are transitive, there 
is a simulation from
$(\Imc_{C,\Omc},d_{C})$ to $(\Imc,d)$. 
Then, by Lemma~\ref{lem:canonical} again,
$d\in C^\Imc$. Since $d$ was an arbitrary element of $\Delta^\Imc$ with $d\in D^\Imc$, this holds for all such elements.
In other words, $D^\Imc\subseteq C^\Imc$.
As \Imc was an arbitrary interpretation satisfying \Omc, we have that $\Omc\models D\sqsubseteq C$. 

Now, suppose $\Omc\models D\sqsubseteq C$.
Then, for all interpretations \Imc satisfying \Omc, we have that $\Imc\models D\sqsubseteq C$. In particular, by Point~2,
$\Imc_{D,\Omc}\models\Omc$. So $\Imc_{D,\Omc}\models D\sqsubseteq C$.
Also, $(\Imc_{D,\Omc},d_D)$ satisfies $D$. That is, $d_{D}\in D^{\Imc_{D,\Omc}}$. As $\Imc_{D,\Omc}\models D\sqsubseteq C$, we have that $d_{D}\in C^{\Imc_{D,\Omc}}$. 
 This finishes the proof of Point~1.

\smallskip

 We now show Point~2. In the following, we assume that $A,B$ are arbitrary concept names in $\NC$
 and $S$ is an arbitrary role name in $\NR$ or its inverse (see Section~\ref{sec:additional} for additional preliminaries and notation). 
\begin{claim} If $\Omc\models A\sqsubseteq B$ then
  $\Imc_{D,\Omc}\models A\sqsubseteq B$.   
\end{claim}
\begin{proof}
        \textbf{Assume $\Omc\models A\sqsubseteq B$.} 
        We make a case distinction based 
on the elements in $\Delta^{\Imc_{D,\Omc}}:=       \Delta^{\Imc_D} \cup \Delta^{{\sf sig}(\Omc)}$. 
\begin{itemize}
\item $d_C\in \Delta^{\Imc_D}$: Assume $d_C\in A^{\Imc_{D,\Omc}}$. In this case, by definition of $\Imc_{D,\Omc}$, either  
 $d_{C}\in A^{\Imc_D}$  
or 
   $\Omc\models C\sqsubseteq A$.
In the former case, by Proposition~\ref{prop:tree}, 
$\models C\sqsubseteq A$. Then in both cases $\Omc\models C\sqsubseteq A$.
By assumption, $\Omc\models A\sqsubseteq B$, so $\Omc\models C\sqsubseteq B$.
Then, again by the definition
of $\Imc_{D,\Omc}$, we have that
$d_C\in B^{\Imc_{D,\Omc}}$.
Since $d_C$ was an arbitrary
element in $\Delta^{\Imc_D}$, this holds   for all  elements of this kind.
\item 
$d_C\in \Delta^{{\sf sig}(\Omc)}$: Assume $d_C\in A^{\Imc_{D,\Omc}}$. Then, by definition of $\Imc_{D,\Omc}$,    $\Omc\models C\sqsubseteq A$.
By assumption, $\Omc\models A\sqsubseteq B$, so $\Omc\models C\sqsubseteq B$.
Then, again by the definition
of $\Imc_{D,\Omc}$, we have that
$d_C\in B^{\Imc_{D,\Omc}}$. 
Since $d_C$ was an arbitrary
element in $\Delta^{{\sf sig}(\Omc)}$, this argument can be applied  for all  elements of this kind.
\end{itemize}
  Then,
    for all elements $d$ in $\Delta^{\Imc_{D,\Omc}}$,
    if $d\in A^{\Imc_{D,\Omc}}$ then $d\in B^{\Imc_{D,\Omc}}$. So
$\Imc_{D,\Omc}\models A\sqsubseteq B$.
%
%
\end{proof}
\begin{claim}\label{claim:canonicalnegation} If $\Omc\models A\sqsubseteq \neg B$ then
  $\Imc_{D,\Omc}\models A\sqsubseteq \neg B$.   
\end{claim}
\begin{proof}
        \textbf{Assume $\Omc\models A\sqsubseteq \neg B$.} 
        We make a case distinction based 
on the elements in $\Delta^{\Imc_{D,\Omc}}:=       \Delta^{\Imc_D} \cup \Delta^{{\sf sig}(\Omc)}$. 
\begin{itemize}
\item $d_C\in \Delta^{\Imc_D}$: Assume $d_C\in A^{\Imc_{D,\Omc}}$. In this case, by definition of $\Imc_{D,\Omc}$, either 
 $d_{C}\in A^{\Imc_D}$  
or 
  $\Omc\models C\sqsubseteq A$.
In the former case, by Proposition~\ref{prop:tree}, 
$\models C\sqsubseteq A$. Then, in both cases $\Omc\models C\sqsubseteq A$.
By assumption, $\Omc\models A\sqsubseteq \neg B$, so $\Omc\models C\sqsubseteq \neg B$. By Proposition~\ref{prop:tree}, $d_C\in C^{\Imc_D}$ (since $\models C\sqsubseteq C$ trivially).
As $\Imc_D$ is a sub-interpretation of
$\Imc_{D,\Omc}$, $d_C\in C^{\Imc_{D,\Omc}}$. 
By assumption, \Omc is satisfiable.
Then, $C^{\Imc_{D,\Omc}}\neq\emptyset$, means that \Omc is satisfiable w.r.t.  $C$. 
So
$\Omc\not\models C\sqsubseteq  B$.
Then, again by the definition
of $\Imc_{D,\Omc}$, we have that
$d_C\notin B^{\Imc_{D,\Omc}}$. That is,
$d_C\in (\neg B)^{\Imc_{D,\Omc}}$.
Since $d_C$ was an arbitrary
element in $\Delta^{\Imc_D}$, this holds   for all  elements of this kind.
\item 
$d_C\in \Delta^{{\sf sig}(\Omc)}$: Assume $d_C\in A^{\Imc_{D,\Omc}}$. Then, by definition of $\Imc_{D,\Omc}$,    $\Omc\models C\sqsubseteq A$.
By assumption, $\Omc\models A\sqsubseteq \neg B$, so $\Omc\models C\sqsubseteq \neg B$.
By definition of $\Delta^{{\sf sig}(\Omc)}$, $C$ is satisfiable w.r.t \Omc.
Then, $\Omc\not\models C\sqsubseteq  B$ and by the definition
of $\Imc_{D,\Omc}$,
$d_C\in (\neg B)^{\Imc_{D,\Omc}}$. 
Since $d_C$ was an arbitrary
element in $\Delta^{{\sf sig}(\Omc)}$, this argument can be applied  for all  elements of this kind.
\end{itemize}
  Then,
    for all elements $d$ in $\Delta^{\Imc_{D,\Omc}}$,
    if $d\in A^{\Imc_{D,\Omc}}$ then $d\in (\neg B)^{\Imc_{D,\Omc}}$. So
$\Imc_{D,\Omc}\models A\sqsubseteq \neg B$.
%
%
\end{proof}


\begin{claim}
If $\Omc\models\exists S\sqsubseteq A$ then $\Imc_{D,\Omc}\models\exists S\sqsubseteq A$.
\end{claim}
\begin{proof}
    \textbf{Assume $\Omc\models\exists S\sqsubseteq A$.} We make a case distinction based 
on the elements in $\Delta^{\Imc_{D,\Omc}}:=       \Delta^{\Imc_D} \cup \Delta^{{\sf sig}(\Omc)}$. 
\begin{itemize}
\item $d_C\in \Delta^{\Imc_D}$: Assume $d_C\in (\exists S)^{\Imc_{D,\Omc}}$. In this case, by definition of $\Imc_{D,\Omc}$, 
either 
there is $d_{C'}\in \Delta^{\Imc_D}$ such that
$(d_{C},d_{C'})\in S^{\Imc_{D}}$   
or 
 there is $d_{\exists \overline{S}}\in  \Delta^{{\sf sig}(\Omc)}$ and $\Omc\models C\sqsubseteq\exists S$.
In the former case, $(d_{C},d_{C'})\in R^{\Imc_{D}}$
implies $d_{C}\in (\exists S)^{\Imc_{D}}$
and
$\models C\sqsubseteq\exists S$, by Proposition~\ref{prop:tree}. Then in both cases $\Omc\models C\sqsubseteq\exists S$.
Since $d_C$ was an arbitrary
element in $\Delta^{\Imc_D}$, this holds   for all such elements.
\item 
$d_C\in \Delta^{{\sf sig}(\Omc)}$: Assume $d_C\in (\exists S)^{\Imc_{D,\Omc}}$. By definition of $\Imc_{D,\Omc}$, there are $2$ possibilities and in both of them  $\Omc\models C\sqsubseteq\exists S$. Since
by assumption  $\Omc\models\exists S\sqsubseteq A$, we have that  $\Omc\models C\sqsubseteq A$. Then,   by definition of $\Imc_{D,\Omc}$, 
$d_C\in A^{\Imc_{D,\Omc}}$.
Since $d_C$ was an arbitrary
element in $\Delta^{{\sf sig}(\Omc)}$, this holds  for all  elements of this kind.
\end{itemize}
  Then,
    for all elements $d$ in $\Delta^{\Imc_{D,\Omc}}$,
    if $d\in (\exists S)^{\Imc_{D,\Omc}}$ then $d\in A^{\Imc_{D,\Omc}}$. So
$\Imc_{D,\Omc}\models\exists S\sqsubseteq A$.
%
\end{proof}

\begin{claim}
If $\Omc\models A\sqsubseteq\exists S$ then $\Imc_{D,\Omc}\models A\sqsubseteq\exists S$.
\end{claim}
\begin{proof}
\textbf{Assume $\Omc\models A\sqsubseteq\exists S$}.
We make a case distinction based 
on the elements in $\Delta^{\Imc_{D,\Omc}}:=       \Delta^{\Imc_D} \cup \Delta^{{\sf sig}(\Omc)}$.
\begin{itemize}
\item $d_C\in \Delta^{\Imc_{D,\Omc}}$: Assume $d_C\in A^{\Imc_{D,\Omc}}$. In this case, by definition of $\Imc_{D,\Omc}$, 
   $\Omc\models C\sqsubseteq A$. 
By assumption, $\Omc\models A\sqsubseteq\exists S$.  
So $\Omc\models C\sqsubseteq\exists S$.  
Then, by definition
of $\Imc_{D,\Omc}$, $(d_C,d_{\exists \overline{S}})\in S^{\Imc_{D,\Omc}}$.
Thus, $d_C\in (\exists S)^{\Imc_{D,\Omc}}$. 
 Since $d_C$ was an arbitrary
element in $\Delta^{\Imc_{D,\Omc}}$, this argument can be applied  for all  elements of this kind.
\end{itemize}
  Then,
    for all elements $d$ in $\Delta^{\Imc_{D,\Omc}}$,
    if $d\in A^{\Imc_{D,\Omc}}$ then $d\in (\exists S)^{\Imc_{D,\Omc}}$. So
$\Imc_{D,\Omc}\models A\sqsubseteq\exists S$.
%
%
\end{proof}



\begin{claim}\label{claim:heredisjoint}
If $\Omc\models\exists S\sqsubseteq \neg A$ then $\Imc_{D,\Omc}\models\exists S\sqsubseteq \neg A$.
\end{claim}
\begin{proof}
    \textbf{Assume $\Omc\models\exists S\sqsubseteq \neg A$.} We make a case distinction based 
on the elements in $\Delta^{\Imc_{D,\Omc}}:=       \Delta^{\Imc_D} \cup \Delta^{{\sf sig}(\Omc)}$. 
\begin{itemize}
\item $d_C\in \Delta^{\Imc_D}$: Assume $d_C\in (\exists S)^{\Imc_{D,\Omc}}$. In this case, by definition of $\Imc_{D,\Omc}$, either 
there is $d_{C'}\in \Delta^{\Imc_D}$ such that
$(d_{C},d_{C'})\in S^{\Imc_{D}}$   
or 
 there is $d_{\exists \overline{S}}\in  \Delta^{{\sf sig}(\Omc)}$ and $\Omc\models C\sqsubseteq\exists S$.
In the former case, $(d_{C},d_{C'})\in R^{\Imc_{D}}$
implies $d_{C}\in (\exists S)^{\Imc_{D}}$
and
$\models C\sqsubseteq\exists S$, by Proposition~\ref{prop:tree}. Then in both cases $\Omc\models C\sqsubseteq\exists S$.
By assumption, $\Omc\models\exists S\sqsubseteq \neg A$. So $\Omc\models C\sqsubseteq \neg A$. 
As argued in the first item of Claim~\ref{claim:canonicalnegation},
$\Omc$ is satisfiable w.r.t. $C$.
So $\Omc\not\models C\sqsubseteq A$.
Then, by definition of $\Imc_{D,\Omc}$,
we have that $d_C\in (\neg A)^{\Imc_{D,\Omc}}$.
Since $d_C$ was an arbitrary
element in $\Delta^{\Imc_D}$, this holds   for all such elements.
\item 
$d_C\in \Delta^{{\sf sig}(\Omc)}$: Assume $d_C\in (\exists S)^{\Imc_{D,\Omc}}$. By definition of $\Imc_{D,\Omc}$, there are two possibilities and in both of them  $\Omc\models C\sqsubseteq\exists S$. Since
by assumption  $\Omc\models\exists S\sqsubseteq \neg A$, we have that  $\Omc\models C\sqsubseteq \neg A$. By definition of $\Delta^{{\sf sig}(\Omc)}$, $C$ is satisfiable w.r.t. \Omc.  Then $\Omc\not\models C\sqsubseteq   A$. By      definition of $\Imc_{D,\Omc}$, 
$d_C\in (\neg A)^{\Imc_{D,\Omc}}$.
As $d_C$ was an arbitrary
element in $\Delta^{{\sf sig}(\Omc)}$, this holds  for all  elements of this kind.
\end{itemize}
  Then,
    for all elements $d$ in $\Delta^{\Imc_{D,\Omc}}$,
    if $d\in (\exists S)^{\Imc_{D,\Omc}}$ then $d\in (\neg A)^{\Imc_{D,\Omc}}$. So
$\Imc_{D,\Omc}\models\exists S\sqsubseteq \neg A$.
%
\end{proof}
The next claim follows from Claim~\ref{claim:heredisjoint}.
\begin{claim}
If $\Omc\models A\sqsubseteq \neg\exists S$ then $\Imc_{D,\Omc}\models A\sqsubseteq \neg\exists S$.
\end{claim}
\end{proof}

 
\lembottomcharacterisation*
\begin{proof}
Let $\Sigma_{\mathsf{C}}'$ be the union of $\Sigma_{\mathsf{C}}$ and the concept names in $\mathsf{sig}(\ont)$ and let $\Sigma_{\mathsf{R}}'$ be the union of $\Sigma_{\mathsf{R}}$ with the role names in $\mathsf{sig}(\ont)$. Further, let $\Phi:=\{A,\exists R,\exists R^-, \mid A\in\Sigma_{\mathsf{C}}',R\in\Sigma_{\mathsf{R}}'\}$ be the finite set of basic concepts over the extended signature $(\Sigma_{\mathsf{C}}',\Sigma_{\mathsf{R}}')$. We build an interpretation $\Imc_{\ont}$ for $(\Sigma_{\mathsf{C}}',\Sigma_{\mathsf{R}}')$ whose domain $\Delta^{\Imc_{\ont}}$ consist of subsets $\Gamma\subseteq\Phi$ that are maximally consistent w.r.t. $\ont$. Since every DL-Lite ontology is satisfiable, the domain will be non-empty.
We set $A^{\Imc_{\ont}}:=\{\Gamma\in\Delta^{\Imc_{\ont}} \mid A\in\Gamma\}$ and $R^{\Imc_{\ont}}:=\{(\Gamma,\Gamma') \mid \exists R\in\Gamma,\exists R^-\in\Gamma'\}$. First, we will show a truth lemma for the maximally consistent sets. Then, we will use that to show that $\Imc_{\ont}\models\ont$ and that every $\lang{\exists,\sqcap,\top,\bot}$ concept that is satisfiable w.r.t. to $\ont$ is satisfied at some point in $\Imc_{\ont}$.

\medskip\par\noindent\emph{Claim (Truth Lemma): $B\in\Gamma$ iff $\Gamma\in B^{\Imc_{\ont}}$}

\medskip\par\noindent\emph{Proof of claim:}
This clearly holds if $B$ is a concept name. Else $B$ is of the form $\exists S$ with $S\in\Sigma_{\mathsf{R}}'$ or $\overline{S}\in\Sigma_{\mathsf{R}}'$.. From right to left, if $\Gamma\in(\exists S)^{\Imc_{\ont}}$, by definition it must be that $\exists S\in\Gamma$. In the other direction, if $\exists S\in\Gamma$ then $\exists S$ and hence also $\exists S^-$ are satisfiable w.r.t. $\ont$. This means that there is some maximally consistent set $\Gamma'$ with $\exists S^-\in\Gamma'$. But then by definition $(\Gamma,\Gamma')\in R^{\Imc_{\ont}}$ and hence $\Gamma\in(\exists S)^{\Imc_{\ont}}$.

\medskip

It follows from the Truth Lemma that $\Imc_{\ont}\models\ont$, since for every concept inclusion $B\sqsubseteq B'\in\ont$ with $B'$ a basic concept, we have $B^{\Imc_{\ont}}=\{\Gamma\in\Delta^{\Imc_{\ont}} \mid B\in\Gamma\}\subseteq\{\Gamma\in\Delta^{\Imc_{\ont}} \mid B'\in\Gamma\}=B'^{\Imc}$ because every $\Gamma\subseteq\Phi$ is maximally consistent w.r.t. $\ont$ and $B,B'\in\Gamma$. Similarly, for a concept inclusion $B\sqsubseteq\neg B'\in\ont$ we have $B^{\Imc_{\ont}}=\{\Gamma\in\Delta^{\Imc_{\ont}} \mid B\in\Gamma\}\subseteq\{\Gamma\in\Delta^{\Imc_{\ont}} \mid B'\not\in\Gamma\}=(\neg B')^{\Imc}$ because $B,B'\in\Phi$ and $\Gamma\subseteq\Phi$ is maximally consistent w.r.t. $\ont$. 

Finally, let $C$ be a $\lang{\exists,\sqcap,\top,\bot}[\Sigma_{\mathsf{C}},\Sigma_{\mathsf{R}}]$-concept that is satisfiable w.r.t. $\ont$. We show that there is some $\Gamma\in\Delta^{\Imc_{\ont}}$ with $\Gamma\in C^{\Imc_{\ont}}$. To show this, we will prove the stronger claim:

\medskip\par\noindent\emph{Claim:}  For every $D\in\lang{\exists,\sqcap,\top,\bot}[\Sigma_{\mathsf{C}},\Sigma_{\mathsf{R}}]$ and basic concept $B\in\Phi$, if $D\sqcap B$ is satisfiable w.r.t.~$\ont$ then $(D\sqcap B)^{\Imc_{\ont}}\neq\emptyset$.

\medskip\par\noindent\emph{Proof of claim:} We prove this by induction on $\mathsf{dp}(D)$.
If $\mathsf{dp}(D)=0$ then the set of all conjuncts of $D\sqcap B$ is a subset of some $\Gamma\in\Delta^{\Imc_{\ont}}$ and hence by the Truth Lemma $\Gamma\in(D\sqcap B)^{\Imc_{\ont}}$. 

For the inductive step, let $D=A_1\sqcap\ldots A_k\sqcap\exists R_1.C_1\sqcap\ldots\exists R_n.C_n$ such that $D\sqcap B$ is satisfiable w.r.t. $\ont$. It follows that $\{A_1,\ldots,A_k,\exists R_1,\ldots,\exists R_n,B\}$ is a consistent set w.r.t. $\ont$, and hence contained in some maximally consistent set $\Gamma\in\Delta^{\Imc_{\ont}}$. By the Truth Lemma, it follows that $\Gamma\in(A_1\sqcap\ldots A_k\sqcap B)^{\Imc_{\ont}}$. 
Moreover, since $D$ is satisfiable w.r.t. $\ont$, so are all concepts of the form $C_i\sqcap\exists R_i^-$ where $1\leq i\leq n$. Therefore, the inductive hypothesis tells us there are types $\Gamma_1,\ldots,\Gamma_n\in\Delta^{\Imc_{\ont}}$ such that $\Gamma_i\in (C_i\sqcap\exists R_i^-)^{\Imc_{\ont}}$ for all $1\leq i\leq n$. Finally, by the Truth Lemma $\exists R_i\in\Gamma_i$ and hence $(\Gamma,\Gamma_i)\in R_i^{\Imc_{\ont}}$ for all $1\leq i\leq n$, which shows that $\Gamma\in(\exists R_1.C_1\sqcap\ldots\exists R_n.C_n)^{\Imc_{\ont}}$ as well, so $\Gamma\in(D\sqcap B)^{\Imc_{\ont}}$.

\medskip
Thus, for any DL-Lite ontology $\ont$, if we set $E^+:=\emptyset$ and $E^-:=\{(\Imc_{\ont},\Gamma) \mid \Gamma\in\Delta^{\Imc_{\ont}}\}$ following the construction above, it follows that $(E^+,E^-)$ is a finite characterisation of $\bot$ w.r.t. $\lang{\exists,\sqcap,\top,\bot}[\Sigma_{\mathsf{C}},\Sigma_{\mathsf{R}}]$ under $\ont$.

\end{proof}
\begin{lemma}\label{lem:subsumption}
    Let $\Omc$ be a DL-Lite TBox in named form, and let $C,C'$ be  $\lang{\sqcap,\exists,\top}$-concept expressions that are satisfiable w.r.t. \Omc.
    Then the following are equivalent:
    \begin{enumerate}
        \item $\Omc\models C\sqsubseteq C'$ and 
        \item there is a simulation
        $Z:(\Imc_{C',\Omc},d_{C'}) \mathop{\underline{\to}} (\Imc_{C,\Omc},d_{C})$,
    \end{enumerate}
    where $\Imc_{C,\Omc}$ and $\Imc_{C',\Omc}$ are as in Definition~\ref{def:can-model}.
\end{lemma}

\begin{proof}
Suppose $\Omc\models C\sqsubseteq C'$. 
By Point~1 of Theorem~\ref{thm:canmod}, 
$d_C\in C'^{\Imc_{C,\Omc}}$.
By Lemma~\ref{lem:canonical}, 
there is a simulation $Z:(\Imc_{C',\Omc},d_{C'}) \mathop{\underline{\to}} (\Imc_{C,\Omc},d_{C})$.
Now, suppose there is a simulation $Z:(\Imc_{C',\Omc},d_{C'}) \mathop{\underline{\to}} (\Imc_{C,\Omc},d_{C})$. By Lemma~\ref{lem:canonical}, $d_C\in C'^{\Imc_{C,\Omc}}$. By Point~1 of Theorem~\ref{thm:canmod}, we have that $\Omc\models C\sqsubseteq C'$.
\end{proof}
\thmpositiveontology*
\begin{proof}
  Let $\ont$ be a DL-Lite ontology and
  let $C$ be a concept in $\lang{\exists,\sqcap,\top,\bot}$. If $C$ or \Omc is unsatisfiable then by Lemma~\ref{lem:bottomcharacterisationontology} it has a finite characterisation w.r.t.~$\lang{\exists,\sqcap,\top,\bot}$ (of exponential size).
  
  Now suppose that $C$ is satisfiable w.r.t.~$\Omc$ and compute the frontier $\Fmc_{\ont} (C):=\{C_1, \ldots, C_n\}$ of $C$ w.r.t.~$\Omc$ in polynomial time using Lemma~\ref{lem:frontiers-el-dllite}. It follows that $|\Fmc_{\ont}(C)|$ and each $|C_i|$ for $C_i\in\Fmc_{\ont}(C)$ are polynomial in $|C|$. Let $E^+ = \{(\Imc_{C,\Omc}, d_{C})\}$ and let
  $E^- = \{(\Imc_{C_i,\Omc}, d_{C_i})\mid 1\leq i \leq n\}$, where $(\Imc_{C_{i},\Omc},d_{C_{i}})$ is the canonical model of $C_{i}$ w.r.t. \Omc. Clearly $(E^+,E^-)$ can be computed from $C$ in polynomial time. We claim that $(E^+,E^-)$ is a finite characterisation for $C$
  w.r.t.~$\Omc$. 
  By Theorem~\ref{thm:canmod}, $E^+$ and $E^-$ consist of pointed interpretations satisfying $\Omc$.  
  
  By Theorem~\ref{thm:canmod}, 
  $C$ is satisfied by the positive example $(\Imc_{C,\Omc}, d_{C})$. 
  For each negative example $(\Imc_{C_i,\Omc}, d_{C_i})\in E^-$, we know that 
  $C_i$ belongs to the frontier $\mathcal{F}$,
  and therefore $\Omc\not\models C_i\sqsubseteq C$. Then, by Theorem~\ref{thm:canmod},
  $C$ is not satisfied by $(\Imc_{C_i,\Omc}, d_{C_i})$. That is, $C$ fits $(E^+,E^-)$.
  
  Finally, let $C'$ be any $\lang{\exists,\sqcap,\top}$-concept expression that fits $(E^+,E^-)$. We 
  show that $ C\equiv_\Omc C'$. 
  As $C'$ is satisfied by the positive example in
  $E^+$, it follows by Theorem~\ref{thm:canmod}
  $\Omc\models C\sqsubseteq C'$.
  
  For the other direction, we proceed by contradiction:
  suppose that
  $\Omc\not\models C'\sqsubseteq C$. By Lemma~\ref{lem:frontiers-el-dllite}, $\Omc\models C_i\sqsubseteq C'$
  for some $C_i\in\mathcal{F}$. 
  Therefore, by Theorem~\ref{thm:canmod},
  $d_{C_i}\in C'^{\Imc_{C_i,\Omc}}$, which means that
  $C'$ fails to fit $(E^+,E^-)$, since $(\Imc_{C_i,\Omc}, d_{C_i})\in E^-$, a contradiction.
\end{proof}

\thmcountingontologies*
\begin{proof}
Note that $A^\Imc=\emptyset$ for all \Imc with $\Imc\models\ont$.
We can then use the same argument  as in Theorem \ref{thm:countingbot}, which shows that $\lang{\geq,\bot}$ does not admit finite characterisations, substituting $A$ by $\bot$.
The argument for  $\lang{\forall,\exists,\sqcap}$ is similar, using the fact that $\lang{\forall,\exists,\sqcap,\bot}$ does not admit finite characterisations under the empty ontology (Theorem~\ref{thm:mainone}).
\end{proof}

\inverseontology*

\begin{proof}
    Observe that $\ont$ entails $A\sqsubseteq \exists R.A$ but does not entail $A\sqsubseteq \exists R^-.A$. 
    Every finite interpretation $(\Imc,d)$ satisfying the ontology $\ont$ that is a positive example
    for  $A$ must contain an infinite outgoing path and hence
    a ``lasso''
    \[ d=d_1\xrightarrow{R} d_2 \cdots \xrightarrow{R} d_m\xrightarrow{R} \cdots \xrightarrow{R} d_{m+k} \xrightarrow{R} d_m\]
    where $d_i\in A^\Imc$ holds for all $i\leq m+k$.
    Note that, in such an interpretation, $d$ also satisfies
    the concept
    expression 
    $C_n := A\sqcap [\exists R.]^n [\exists R^-.]^{n+1} A$ 
    for all $n\geq m$.
    This means that every positive example for $A$ is, for sufficiently large $n$, also a positive example of the concept
    expression $C_n$.
    Furthermore, clearly, every negative example of $A$ is also a 
    negative example of $C_n$. Thus, for any finite set of 
    positive and negative examples for $A$, we have that
    $C_n$ also fits then, if we choose $n$ large enough.
    However, $A$ and $C_n$ are not equivalent relative to $\ont$, because for the interpretation $\mathcal{J}_n$ consisting of a length $n+1$ $R$-path with a loop at the final point
    \[e_1\xrightarrow{~R~}\cdots\xrightarrow{~R~}e_{n}\xrightarrow{~R~}e_{n+1}{\rotatebox{90}{$\circlearrowleft$}}^{R}\]
    where $e_j\in A^{\mathcal{J}_n}$ for all $1\leq j\leq {n+1}$ we have $\mathcal{J}_n\models\ont$ and $e_1\in(A\sqcap\neg C_n)^{\mathcal{J}_n}$ because $e_1$ does not have any $R$-predecessors.
\end{proof}

\thmnegeldlliteh*
\begin{proof}
The proof of Theorem~\ref{thm:inverseontology} can be adapted by replacing $R^-$ by $S$ in the concepts in question.\end{proof}

\section{Proofs of Main Theorems}\label{app:maintheorems}


\mainone*
\begin{proof}
By Thm.~\ref{thm:knownresults} we know that $\{\exists,-,\sqcap,\sqcup,\top,\bot\}$ and $\{\forall,\exists,\sqcap,\sqcup\}$ admit finite characterisations, citing results known in the literature. Further, we show in Thm.~\ref{thm:ALEQcharacterisations} that $\lang{\forall,\exists,\geq,\sqcap,\top}$ admits finite characterisations. Together these prove point 1. For point 2, let $\Obf\subseteq\{\forall,\exists,\geq,-,\sqcap,\sqcup,\top,\bot,\neg\}$ such that it is not the case that $\{\geq,-,\sqcap\}\subseteq\Obf\subseteq\{\exists,\geq,-,\sqcap,\top\}$ and suppose that $\lang{\Obf}$ admits finite characterisations. We show that either $\Obf\subseteq\{\exists,-,\sqcap,\sqcup,\top,\bot\}$, $\Obf\subseteq\{\forall,\exists,\sqcap,\sqcup\}$ or $\Obf\subseteq\{\forall,\exists,\geq,\sqcap,\top\}$.

As $\{\exists,\sqcap\}\subseteq\Obf$, if $\neg\in\Obf$ as well then $\lang{\Obf}$ can express every $\mathcal{ALC}$ concept. However, it follows from Thm.~\ref{thm:knownresults} that $\mathcal{ALC}$ does not admit finite characterisations. So suppose that $\neg\not\in\Obf$, i.e. $\Obf\subseteq\{\forall,\exists,\geq,-,\sqcap,\sqcup,\top,\bot\}$. If $\Obf\not\subseteq\{\exists,-,\sqcap,\sqcup,\top,\bot\}$ it must be that $\forall\in\Obf$ or $\geq\in\Obf$ (or both). Thus if $\geq\not\in\Obf$ it must be that $\forall\in\Obf$ and hence $\{\forall,\exists,\sqcap\}\subseteq\Obf\subseteq\{\forall,\exists,-,\sqcap,\sqcup,\top,\bot\}$. However, by Thm.~\ref{thm:inversecap} we know that $\lang{\forall,\exists,-,\sqcap}$ does not admit finite characterisations and by Thm.~\ref{thm:knownresults} also $\lang{\forall,\exists,\sqcap,\bot}$ and $\lang{\forall,\exists,\sqcup,\top}$ do not admit finite characterisations. It follows that either $\Obf\subseteq\{\forall,\exists,\sqcap,\top\}\subseteq\{\forall,\exists,\geq,\sqcap,\top\}$ or $\Obf\subseteq\{\forall,\exists,\sqcap,\sqcup\}$.

Let us now consider the cases where $\geq\in\Obf$, i.e. $\{\exists,\geq,\sqcap\}\subseteq\Obf\subseteq\{\forall,\exists,\geq,-,\sqcap,\sqcup,\top,\bot\}$. By Thm.~\ref{thm:countingbot} we know that respectively $\lang{\geq,\bot},\lang{\geq,\sqcup}$ do not admit finite characterisations. It follows that in fact $\Obf\subseteq\{\forall,\exists,\geq,-,\sqcap,\top\}$. If $\forall\in\Obf$, then by Thm.~\ref{thm:inversecap} it cannot be that $-\in\Obf$ as well, so $\Obf\subseteq\{\forall,\exists,\geq,\sqcap,\top\}$. Else $\forall\not\in\Obf$ so $\{\exists,\geq,\sqcap\}\subseteq\Obf\subseteq\{\exists,\geq,-,\sqcap,\top\}$. Since we assumed that $\Obf$ is not between $\{\geq,-,\sqcap\}$ and $\{\exists,\geq,-,\sqcap,\top\}$, it follows that $-\not\in\Obf$ so $\Obf\subseteq\{\exists,\geq,\sqcap,\top\}\subseteq\{\forall,\exists,\geq,\sqcap,\top\}$.
\end{proof}

\maintwo*\begin{proof}
By Thm.~\ref{thm:knownresults} we know that $\lang{\exists,-,\sqcap,\top,\bot}$ admits polynomial-time computable characterisations. Now suppose that $\{\exists,\sqcap\}\subseteq\Obf\subseteq\{\forall,\exists,\geq,-,\sqcap,\sqcup,\top,\bot,\neg\}$ is such that $\lang{\Obf}$ admits polynomial characterisations. We show that $\Obf\subseteq\{\exists,-,\sqcap,\top,\bot\}$.

First of all, note that it cannot be that $\neg\in\Obf$ since, as noted before, it follows from Thm.~\ref{thm:knownresults} that $\lang{\exists,\sqcap,\neg}$ does not admit finite characterisations, so $\Obf\subseteq\{\forall,\exists,\geq,-,\sqcap,\sqcup,\top,\bot\}$. Further, it follows from Thm.~\ref{thm:ELQlowerbound} that it cannot be that $\geq\in\Obf$ and it follows from Corollary~\ref{cor:ALE_NPhardness} that it cannot be that $\forall\in\Obf$ (unless $P\ne NP$). Moreover, Thm.~\ref{thm:knownresults} states that $\lang{\exists,\sqcap,\sqcup}$ does not admit polynomial time computable characterisations. Thus, it also cannot be that $\sqcup\in\Obf$ and hence $\Obf\subseteq\{\exists,-,\sqcap,\top,\bot\}$ as desired. 
\end{proof}

\mainthree*


\begin{proof}
Theorem \ref{thm:EL-poly-char-DLLite} shows that $\lang{\exists,\sqcap,\top,\bot}$ admits finite characterisations under DL-Lite ontologies. In fact, it was shown there that for every DL-Lite ontology $\ont$ and concept $C\in\lang{\exists,\sqcap,\top,\bot}$ that is satisfiable w.r.t.~$\ont$, a finite characterisation for $C$ w.r.t.~$\mathcal{L}$ under $\ont$ can be computed in polynomial time.

Now suppose that $\lang{\Obf}$ admits finite characterisations under DL-Lite ontologies, where $\Obf\subseteq\{\forall,\exists,\geq,-,\sqcap,\top,\bot,\neg\}$. We show that $\Obf\subseteq\{\exists,\sqcap,\top,\bot\}$. It follows again from Thm.~\ref{thm:knownresults} that it cannot be that $\neg\in\Obf$. Further, Theorem \ref{thm:inverseontology} shows that $-\not\in\Obf$, while Thm.~\ref{thm:countingontologies} shows that $\geq\not\in\Obf$ and $\forall\not\in\Obf$. Hence, it follows that $\Obf\subseteq\{\exists,\sqcap,\top,\bot\}$.
\end{proof}

\section{Repercussions for Learnability}
\label{app:learning}

We will define what a membership query learning algorithm is for learning description logic concepts without a background ontology. 
For a general definition of membership query learning algorithm, we refer the reader to \cite{AngluinQueriesConcepts}.

\begin{definition}[Membership Query Learning Algorithms]
Let $\mathcal{L}$ be a concept language. An \emph{exact learning algorithm with membership queries} for $\mathcal{L}$ 
takes as input finite sets $\Sigma_\mathsf{C}$ and $\Sigma_\mathsf{R}$ 
and has to learn a hidden target concept $C\in\mathcal{L}[\Sigma_\mathsf{C},\Sigma_\mathsf{R}]$ by making \emph{membership queries} to a membership oracle. In a membership query, the learner presents an example $(\Imc,d)$ 
to the oracle, to which the oracle faithfully has to answer ``yes'' if $d\in C^{\Imc}$ and ``no'' otherwise. After asking finitely many such membership queries, the algorithm terminates and
outputs a concept expression $D\in\mathcal{L}[\Sigma_\mathsf{C},\Sigma_\mathsf{R}]$ such that $C\equiv D$.
We say that $\mathcal{L}$ is \emph{membership query learnable}
if there exists an exact learning algorithm with membership queries for $\mathcal{L}$ 
Moreover, it is \emph{polynomial time membership query learnable} if the algorithm runs in polynomial time.
\end{definition}

Recall our discussion on the differences between the ABox-as-examples and interpretations-as-examples setting (Remark~\ref{rem:exampledef}). This difference carries through in the definition of membership queries, and hence our definition of exact learning differs from the one in e.g.~\cite{DuplicateDBLP:conf/ijcai/FunkJL22}.

The following result says that (polynomial time computable) finite characterisations are a necessary condition for (polynomial time) membership query learnability.

\begin{theorem}\label{thm:learningneccond}
If a concept language $\mathcal{L}$ is (polynomial time) membership query learnable 
and $\mathcal{L}$ has a polynomial time model checking problem in combined complexity, then $\mathcal{L}$ admits (polynomial time computable) finite characterisations.
\end{theorem}
\begin{proof}
Let $\mathcal{A}$ be the query learning algorithm using only membership queries to identify concept expressions $C\in\mathcal{L}$ in polynomial time. For any concept $C\in\mathcal{L}$, we will generate in polynomial time a finite characterisation of $C$ w.r.t. $\mathcal{L}$ as follows. 

Run the learning algorithm $\mathcal{A}$ on input $C$, and answer membership queries in polynomial time by using the model checking algorithm which runs in polynomial time. We record all the examples visited during the run, and label them as positive or negative depending on the answer of the membership query. The resulting finite set of labelled examples, the \emph{trace} of $\mathcal{A}$ on input $C$, must be a finite characterisation of $C$ w.r.t. $\mathcal{L}$. For suppose otherwise, then there must be some other non-equivalent concept $C'\in\mathcal{L}$ with $C\not\equiv C'$ that is consistent with all labelled examples seen so far. But that means that $\mathcal{A}$ will never output $C'$ as hypothesis, and hence will never be able to learn $C'$, contradicting our assumption that $\mathcal{A}$ is a membership query learning algorithm. Moreover, if $\mathcal{A}$ runs in polynomial then this whole construction of finite characterisations above is also in polynomial time. 
\end{proof}

We obtain a complete classification of which fragments of $\mathcal{ALCQI}$ containing $\sqcap$ and $\exists$ membership query learning, analogous to Main Theorem~\ref{thm:mainone}

\begin{corollary}
\mbox{Let $\{\exists,\sqcap\}\subseteq \Obf\subseteq\{\forall,\exists,\geq,-,\sqcap,\sqcup,\top,\bot,\neg\}$.} 
\begin{enumerate}
    \item If $\Obf$ is subsumed by $\{\exists,-,\sqcap,\top,\bot\}$, then $\lang{\Obf}$ is polynomial-time membership query learnable.
    \item  Otherwise, $\lang{\Obf}$ is not polynomial-time membership query learnable, assuming $\mathsf{P}\ne\mathsf{NP}$.
\end{enumerate}
\end{corollary}
\begin{proof}
From right to left, a polynomial time membership query learning algorithm for an extension of $\lang{\exists,-,\sqcap,\top,\bot}$ was given in \cite{BalderCQ}. The other direction is immediate from Theorems \ref{thm:maintwo}, \ref{thm:learningneccond} and \ref{thm:modelchecking}.
\end{proof}

Conversely, the existence of finite characterisations guarantees the existence of an inefficient naive membership query learning algorithm. However, the mere existence of polynomial time computable characterisations does \emph{not} imply the existence of a polynomial time membership query learning algorithm.

\begin{theorem}\label{thm:learningsuffconf}
If a concept language $\mathcal{L}$ admits computable characterisations then $\mathcal{L}$ is (not necessarily efficiently) membership query learnable 
\end{theorem}

\begin{proof}
Our naive learning algorithm works as follows. Fix some enumeration of all the concepts in $\mathcal{L}$, i.e. $\mathcal{L}=\{C_0,C_1,C_2,\ldots\}$,
and for each  $C_i$ compute a finite characterisation $(E^+,E^-)$ w.r.t.~$\mathcal{L}$ 
and ask a membership query for every example in $E^+\cup E^-$. If the answer to the membership queries fits the labeling, output the concept $C_i$. Otherwise, continue to the next concept in the enumeration.
\end{proof}

Thus, the following corollary is immediate from Theorems~\ref{thm:mainone}, \ref{thm:learningneccond} and \ref{thm:learningsuffconf}.

\begin{corollary}
\mbox{Let $\{\exists,\sqcap\}\subseteq \Obf\subseteq\{\forall,\exists,\geq,-,\sqcap,\sqcup,\top,\bot,\neg\}$.}
\begin{enumerate}
    \item If $\Obf$ is subsumed by $\{\exists,-,\sqcap,\sqcup,\top,\bot\}$, $\{\forall,\exists,\sqcup,\sqcap\}$ or $\{\forall,\exists,\geq,\sqcap,\top\}$ then $\lang{\Obf}$ is membership query learnable.
\item Otherwise $\lang{\Obf}$ is not membership query learnable, except possibly if
 $\{\geq,-,\sqcap\}\subseteq\Obf\subseteq\{\exists,\geq,-,\sqcap,\top\}$. 
\end{enumerate}
\end{corollary}

\end{document}